\documentclass[11pt]{article}

\usepackage{array}
\usepackage{amsfonts,amsmath,amssymb,amsthm,mathrsfs}
\usepackage{adjustbox}
\usepackage{tikz}
\usepackage{enumitem}
\usepackage{ifthen}
\usepackage{fullpage}

\usepackage[ruled]{algorithm2e} 

\usepackage{algorithmic}
\SetAlFnt{\small}
\SetAlCapFnt{\small}
\SetAlCapNameFnt{\small}
\SetAlCapHSkip{0pt}
\IncMargin{-\parindent}

\usepackage{booktabs} 
\usepackage{xspace}
\usepackage[pagebackref=true,hidelinks]{hyperref}
\usepackage{bbm}
\usepackage{enumitem}
\usepackage[numbers]{natbib}
\usepackage{multirow}
\usepackage{graphicx}
\usepackage{subfigure}
\usepackage{wrapfig}
\usepackage{cleveref}
\usepackage{comment}
\usepackage{thm-restate}

\newcommand\x{\vec x}

\newcommand{\f}{f}

\newcommand\opt{\mathsf{opt}}

\newtheorem{theorem}{Theorem}[section]
\newtheorem{lemma}[theorem]{Lemma}
\newtheorem*{theorem*}{Theorem}
\newtheorem*{lemma*}{Lemma}
\newtheorem{definition}[theorem]{Definition}
\newtheorem{corollary}[theorem]{Corollary}
\newtheorem{observation}[theorem]{Observation}

\newtheorem{claim}[theorem]{Claim}

\newtheorem{optimization}[theorem]{Optimization Problem}

\newcommand{\eps}{\epsilon}

\newcommand{\alg}{\textsc{ALG}}
\newcommand{\E}{\mathbb{E}}
\renewcommand{\opt}{\textsc{OPT}}
\newcommand{\ub}{\textsc{UB}}
\newcommand{\lb}{\textsc{LB}}
\newcommand{\ratio}{\textsc{R}}
\newcommand{\si}{SR }

\newcommand{\dist}{D}
\newcommand{\adx}{AdEx }
\newcommand{\gam}{\gamma}

\begin{document}

\title{Online Allocation and Display Ads Optimization with Surplus Supply}
\author{
Melika Abolhassani \\ Google \\ {\tt melikaa@google.com} \and 
Hossein Esfandiari \\ Google \\ {\tt esfandiari@google.com} \and 
Yasamin Nazari \\ Johns Hopkins University \\ {\tt ynazari@jhu.edu} \and 
Balasubramanian Sivan \\ Google \\ {\tt balusivan@google.com} \and
Yifeng Teng \\ UW-Madison \\ {\tt yifengt@cs.wisc.edu} \and 
Creighton Thomas \\ Google  \\ {\tt creighton@google.com}
}
\date{}

\maketitle
\thispagestyle{empty}

\begin{abstract}
In this work, we study a scenario where a publisher seeks to maximize its total revenue across two sales channels: guaranteed contracts that promise to deliver a certain number of impressions to the advertisers, and spot demands through an Ad Exchange. On the one hand, if a guaranteed contract is not fully delivered, it incurs a penalty for the publisher. On the other hand, the publisher might be able to sell an impression at a high price in the Ad Exchange. How does a publisher maximize its total revenue as a sum of the revenue from the Ad Exchange and the loss from the under-delivery penalty? We study this problem parameterized by \emph{supply factor $f$}: a notion we introduce that, intuitively, captures the number of times a publisher can satisfy all its guaranteed contracts given its inventory supply. In this work we present a fast simple deterministic algorithm with the optimal competitive ratio. The algorithm and the optimal competitive ratio are a function of the supply factor, penalty, and the distribution of the bids in the Ad Exchange. 

Beyond the yield optimization problem, classic online allocation problems such as online bipartite matching of Karp-Vazirani-Vazirani~\cite{KVV90} and its vertex-weighted variant of Aggarwal et al.~\cite{AGKM11} can be studied in the presence of the additional supply guaranteed by the supply factor. We show that a supply factor of $f$ improves the approximation factors from $1-1/e$ to $f-fe^{-1/f}$. Our approximation factor is tight and approaches $1$ as $f \to \infty$.
\end{abstract}

\newpage
\setcounter{page}{1}

\section{Introduction}
\label{sec:intro}
An overwhelming majority of publishers on the web monetize their service by displaying ads alongside their content. The revenue stream of such publishers typically comes from two key channels, often referred to as direct sales and indirect sales. In the direct sales channel the publisher strikes several contracts with some major advertisers. The price of such contracts are often negotiated and decided on a per-impression basis before the serving begins. 
In the indirect sales channel, the ad is selected by seeking, in real-time, bids in an Ad Exchange platform (\adx for short). In this case an auction is conducted to select the winner and decide how much they pay. 
A comprehensive yield optimization consists of jointly optimizing the publisher's revenue across both channels. In fact, revenue optimization in this context is significantly important since the display ads industry represents a giant ($>$ \$50B) marketplace and is fast growing even at its current mammoth size.

\paragraph{Basic setting and preliminaries.} We begin by formally describing our setting. The joint yield optimization problem can be modeled as an online edge-weighted and vertex-capacitated bipartite matching problem. There is a set $A$ of offline vertices that correspond to the advertisers with contracts (direct sales), and there is an additional special offline vertex $a_d$ representing \adx (indirect sales). Advertiser $a \in A$ has capacity $n_a$ and we have $n_{a_d} = \infty$. The capacity $n_a$ represents the number of impressions demanded by contractual advertiser $a$. Let $N = \sum_{a \in A} n_a$. There is a penalty $c$ that the publisher pays an advertiser for every undelivered impression\footnote{We later discuss relaxing the penalty $c$ to depend on the advertiser $a$.}: i.e., if at the end of the algorithm we assign $k_a < n_a$ impressions to $a \in A$, the publisher pays $c(n_a-k_a)$ to $a$ (there is no benefit to the publisher for delivering beyond $n_a$ impressions). The publisher is not obligated to deliver any impression to \adx, and thus doesn't incur any penalty from $a_d$. Advertisers are represented as \textit{offline vertices}. Users/queries, arrive \textit{online} in an adversarial manner, and they constitute the online vertex set. When an online vertex (query) arrives, the set of its incident edges to offline vertices (representing the offline nodes that are eligible to be assigned this query) becomes known to the algorithm. \emph{Every} arriving query has an edge to the \adx node $a_d$, i.e., every query can be sent to an exchange seeking a bid. All edges incident on any node $a \in A$ have the same weight\footnote{Unweighted edges for contractual advertisers is fine because these contracts are mostly based on the number of impressions delivered. In a few cases the contracts are based on the number of clicks or conversions, in which case the edges will be weighted based on the probability of click or conversion. Contracts based on impressions form such a large majority, that having unweighted edges, is almost wlog.} and the edges incident on the \adx node $a_d$ could have an arbitrary weight depending on the highest bid from the Exchange. \adx is modeled by the distribution $D$ of highest bids in the exchange: i.e., regardless of the query that arrives, when it is assigned to $a_d$, the publisher accrues a profit that is equal to a draw from $D$. The publisher's basic problem is to decide, on a per-query basis, whether to assign the query to a contract advertiser (if so, whom) or to \adx. 

Publisher's goal is to maximize its overall revenue. Publishers typically have pre-negotiated prices $p_a$ for each contractual advertiser $a$. The total revenue of the publisher will be the sum of three parts (i) the revenue from \adx (i.e., the sum of edge weights of queries assigned to \adx), (ii) the revenue from contracts: $\sum_{a\in A}n_a \cdot p_a$, and (iii) the revenue lost due to under-delivery, i.e., the negative of the penalty paid. Note that (ii) is a constant, and is unaffected by the allocation algorithm. Thus, while computing competitive ratio, we compute it w.r.t. the sum of (i) and (iii).  

\paragraph{Supply factor.} An important concept that we introduce is what we call a \textit{supply factor} of an instance, which captures the (potentially fractional) number of times that a publisher will be able to satisfy their contractual advertisers' demands. Formally, let a complete matching be defined as one where all contractual advertisers' demands $n_a$ are fully satisfied, i.e., all the offline vertices are fully saturated. The supply factor of an instance is defined as the largest positive real number $f$ s.t., there exists an \emph{offline solution} with $f$ complete matchings. If there are many such matchings, we pick one to be the supply-factor-determining-offline-solution. In this work, we assume that the number of arriving online queries is exactly $fN = f\sum_{a \in A} n_a$. The algorithm designer is aware of $f$, the $n_a$'s, and the highest bid distribution from \adx.


There are several important practical aspects of the yield optimization problem that previous work do not capture that we aim to address:
\begin{enumerate}[align=left]
\item The first aspect is that publishers typically have more inventory than they are able to sell via the direct sales channel (contracts), and indeed that is the main reason that most publishers are selling through the indirect sales channel of \adx as well. 
Most previous works on joint yield optimization either address the objectives of the two channels separately (bi-criteria objective), or study them in the absence of supply factor/penalties/\adx bid distribution. Studying the yield optimization problem with a single unified objective (\adx revenue - penalty) in the presence of supply factor and \adx bid distribution surfaces the nature of the optimal tradeoff between the supply factor and how on-track a contract is towards hitting its goals. Clearly, when a contract is lagging behind, we should allocate a query to \adx only when the \adx bid is high enough. But how does this ``high enough" vary as we increase/decrease the publisher's supply, captured by the supply factor $f$? This is explicitly answered in our work. Similarly the dependence on the penalty and \adx distribution are also explicitly revealed.


\item Even in classic online allocation problems like the online bipartite matching of Karp et al.~\cite{KVV90} and the online vertex-weighted bipartite matching of Aggarwal et al.~\cite{AGKM11}, it is interesting to inquire what happens to the competitive ratio when there is a supply factor $f \geq 1$.

\item Prior works mostly studied the problem in a fully stochastic model or a fully adversarial model. In reality, while user browsing patterns might have significant variations across days, in response to events, state-of-mind etc. (and hence an adversarial arrival of queries is reasonable), advertiser bidding/spending patterns are far more predictable because advertisers have daily and hourly spending budgets. We incorporate this in our model by having a distribution $D$ over the highest bids from \adx, even though query arrival is adversarial. The inclusion of \adx bid distribution, not only represents reality better, but also leads to a crisp algorithm that sheds ample light on the role of the distribution in the joint yield optimization problem. 

\end{enumerate}



\subsection{Our Results} 
One of our contributions, as just discussed, is to present an economical model that crisply captures the reality of display ads monetization. Our main result is a fast simple deterministic algorithm that obtains the optimal competitive ratio as $n_a$  values grow large. 
The algorithm is as follows: let $0=r_1 < \dots < r_d$ be the points in the support of the distribution $D$ of highest bid in \adx (highest bid is often referred to as reward for short). As a pre-processing step, compute $d$ thresholds $s_1 < \dots < s_d$ as a function of $f$ (we define $s_0 = 0$ and $s_{d+1} = 1$), $c$ and the \adx bid distribution. Let the satisfaction-ratio $\si(a)$ of a contractual advertiser $a$ be the ratio of the number of impressions delivered to the contract thus far, to the number of impressions $n_a$ requested by the contract. For each arriving query, the algorithm picks the contract with the lowest satisfaction ratio, call it $s$. Find $u$ such that $s \in [s_{u-1}, s_u)$. Assign the query to \adx if the highest bid $r$ in the exchange exceeds $r_{d+1-u}$. And if not, assign the query to the contract with the lowest satisfaction ratio. Algorithm~\ref{alg:general} summarizes this. We highlight a few important aspects of this algorithm.

\begin{algorithm}[htb]
\caption{Optimal algorithm for general \adx distribution}
\label{alg:general} \small
\textbf{Input:}{ \adx distribution $\dist$ with support $0 =r_1 < ... < r_d$, penalty $c$, and supply factor $f$.}\\
\textbf{Preprocessing:} Compute thresholds $s_1,...,s_d$ (we discuss how in Optimization Problem \ref{def:lb} ).\\
 \For{each query arriving online }{
    Let $r$ be the highest \adx bid for this query.\\
    Let $a$ be the advertiser, with an edge to this query, and with the lowest satisfaction ratio $\si(a)$.\\
    \If{$\si(a)=1$}{
        Assign the impression to \adx.
      }
    \Else{
    Find $u$ such that $\si(a) \in [s_{u-1}, s_u)$.\\
    \If{$r \leq r_{d+1-u}$}{
        Assign the impression to advertiser $a$.
      }
    \Else{Assign the impression to \adx.}
    }
    }
\end{algorithm}

\begin{enumerate}
\item Once the pre-processing step is over (which is a one-time computation), the algorithm is very simple to implement in real time while serving queries, even in a distributed fashion. Each relevant advertiser $a$ for the current query (i.e., each offline node $a$ with a matching edge to the current online node) just responds with its satisfaction ratio $\si(a)$. From there on, the algorithm simply computes the smallest satisfaction ratio, and do a simple lookup over the thresholds that are pre-computed, and decide the allocation based on how big the \adx bid is.
\item The algorithm is quite intuitive. As the satisfaction ratio of the most needy contract gets lower, the \adx bid has to be correspondingly higher to merit snatching this impression from the contract. This tradeoff happens to take such a simple symmetric form, where one looks for the mirror image in $\vec{r}$, namely $r_{d+1-u}$, of the index $u$ to which the satisfaction ratio gets mapped is quite surprising. Importantly, the supply factor and penalty are used only in the pre-processing step to compute the thresholds, and don’t appear in serving time at all.
\item The algorithm need not fully know the highest bid $r$ from \adx. It just needs to be able to compare the highest bid against a reserve price of $r_{d+1-u}$. Further, extending the algorithm to deal with multiple Ad Exchanges is simple: broadcast the same reserve to all exchanges, and pick the highest bidding exchange that clears the reserve (we just need to know which exchange is the highest bidder, and whether they clear the reserve, not the exact value of the bid). If no exchange clears the reserve, allocate to the advertiser $a$ with the lowest $\si(a)$.
\item While the algorithm is intuitive in hindsight, it is far from obvious that it obtains the optimal competitive ratio. 
\end{enumerate}

As mentioned earlier, apart from analyzing the joint yield optimization problem, we also show the benefits that a supply factor can bring in classic online algorithmic problems. For the seminal online bipartite matching problem of ~\cite{KVV90}, we show that the same RANKING algorithm of ~\cite{KVV90} with a supply factor of $f$ yields a tight competitive ratio of $f - fe^{-1/f}$, which increases with $f$, and approaches $1$ as $f\to\infty$. Likewise for the vertex-weighted generalization of this problem studied by~\cite{AGKM11}, the same generalized vertex-weighted RANKING algorithm of ~\cite{AGKM11} (a.k.a \textsc{Perturbed Greedy}) yields a competitive ratio of $f - fe^{-1/f}$. 
We defer these analyses to the Appendix~\ref{app:uniform}, and include them primarily to show how supply factor influences the competitive ratio of some well known problems.

\paragraph{Overview of analysis techniques.} We use a max-min approach to analyze the performance of our algorithm. Given the thresholds $s_1 < \dots < s_d$, our algorithm is completely defined. Therefore the adversary can compute the instance that minimizes the optimal objective of our algorithm given the thresholds, and the algorithm can optimize the thresholds $s_1 < \dots s_d$ knowing the best response of the adversary. The minimization problem of the adversary can be captured by a succinct LP, and we reason about the structure of the optimal solution to this LP. This sets up the maximization problem of the algorithm, which turns out to be a non-linear, non-convex optimization problem. Nevertheless, we develop a simple poly-time dynamic programming algorithm that obtains the optimal solution (optimal thresholds $s_1,\dots,s_d$) up to a small additive error. For tightness, we construct an example which is a modified version of the ``upper triangular graph'' of Karp et al.~\cite{KVV90}, and show that no algorithm can obtain an objective value larger than the objective value achieved as the optimal solution to the max-min problem described above. This establishes that the class of threshold-based algorithms is optimal. To act as a warm up to ease into the general distribution section, we begin with the special case of distributons with support size two. In this case, the maximization problem of the algorithm in the max-min problem above is a single-variable concave maximization problem, and already yields clear insights on how the optimal threshold computed by the algorithm depends on the supply factor $f$ and the penalty $c$.

\paragraph{Bid-to-budget ratio vs supply factor.} On the surface level, it might appear that the notion of supply factor is just like the ``large budgets'' assumption, where it is assumed that the budget (in our case the number of impressions $n_a$ demanded by each advertiser $a$) is much larger than the bid (i.e., the value of an edge). However these two concepts are quite different. In particular, even with the large budgets assumption, without a supply factor larger than $1$, any algorithm will be very conservative and will essentially always allocate to the contracts (assuming the penalty is larger than the \adx reward). The supply factor is a property of the entire setup of the publisher: the demands of the contracts and the nature of traffic (set of online nodes arriving, i.e., users/queries that visit their website).

\paragraph{Extensions.} A natural question to ask is what happens if the publishers have different under-delivery penalties $c_a$ for different advertisers. To show a proof of concept extension of our results to this setting, we consider the simpler setting of our problem where the \adx rewards are equal to $r$ for every query (i.e., a deterministic distribution $D$), and show how the technique and results extend to handle different $c_a$'s. We conjecture that the same approach extends to the general \adx distributions as well, and leave it as an open problem. In a different direction, in this work, we focus on a deterministic algorithm because of its many virtues when deployed in a production system: the ability to replay and hence debug easily, ex-post fairness, etc. While we show that it achieves the optimal competitive ratio (i.e., even randomized algorithms cannot improve further), this necessarily requires $n_a$ values being large (for a deterministic algorithm to be optimal, large budgets are necessary even for the much simpler $B$-matching problem~\cite{KP00}). In practice, however, large budget assumption essentially always holds, as advertiser contractual demands are much larger than the edge weight of $1$. Nevertheless, one could ask whether one could use randomized algorithms to remove the dependence of $n_a$'s being large. Again, as a proof of concept extension of our results, we show that for the special case where \adx rewards are uniformly equal to $r$ for every query, randomized algorithms can get the same competitive ratio as deterministic ones for any value of $n_a$, not just large ones. 

\paragraph{Comparison to closely related work.} In terms of works that consider joint optimization across the two channels, the closest to ours is that of Dvor{\'a}k and Henzinger~\cite{DH14}, who also consider the objective of maximizing revenue across two channels: the fundamental differences are (a) the absence of a supply factor in their work, (b) they model adversarially both the arrivals and the \adx bids, and (c) they achieve separate approximation factors for each channel as opposed to our approximating the joint unified objective. Equally close is the work of Balseiro et al.~\cite{BFMM14}, who study the same problem, with the differences being (a) the absence of a supply factor, (b) they model stochastically both the arrivals and \adx bids.

Another closely related work is by Devanur and Jain~\cite{DJ12}  in which they consider the adwords problem with concave returns in the objective: while their model can capture penalties, it does not handle the AdEx reward distribution. Our model takes the reward distribution and penalties into account simultaneously. Additionally, the supply factor notion is absent in~\cite{DJ12}.

A number of works consider the optimization problem without the presence of \adx. Feldman et al.~\cite{FKMM09} study the problem with worst case arrivals and achieve a $1-1/e$ competitive ratio as the $n_a$'s grow large. Feldman et al.~\cite{FHKMS10} study the general packing LPs in a random permutation arrival model and show how to achieve $1-\eps$ approximation as the $n_a$'s grow large, and Devanur and Hayes~\cite{DH09} study the related Adwords problem in the same random permutation model to achieve a $1-\eps$ approximation. Agrawal et al.~\cite{AWY14} show how to attain $1-\eps$ for general packing LPs with better convergence rates on how fast $n_a$'s need to go to $\infty$. Devanur et al.~\cite{DJSW19} consider general packing and covering LPs in an i.i.d. model with unknown distribution and achieve even better convergence rates. Agrawal and Devanur~\cite{AD15} study online stochastic convex programming. Mirrokni et al.~\cite{MGZ12} study the Adwords problem and design algorithms that simultaneously perform well for both stochastic and adversarial settings, and Balseiro et al.~\cite{BLM20} do this for generalized allocation problems with non-linear objectives using dual mirror descent. We refer the reader to Choi et al.~\cite{CMBL20} for a literature review on the display ads market as it is too vast to cover in entirety here. The differentiating factors of all these works from ours is that even if these works were to add an \adx node with infinite capacity, (a) they do not consider the supply factor, (b) and they do not have a unified objective.
Another related work by Esfandiari et al.~\cite{EKM2018} considers the allocation problem in a mixed setting, where a fraction of queries arriving are adversarial, and a fraction are stochastic. They then characterize their competitive ratio, by this \textit{prediction fraction}. The setting we consider is different, as we allow fully adversarial queries. We only assume a known \adx distribution, which we argued is often more predictable than the user traffic.

Karp et al.~\cite{KVV90} wrote the seminal paper on online bipartite matching, and Aggarwal et al.~\cite{AGKM11} consider the generalization of it to vertex weighted settings. Mehta et al.~\cite{MSVV07} introduced the influential Adwords problem and gave a $1-1/e$ approximation for it, with a recent breakthrough result by Huang et al.~\cite{HZZ20} showing how to beat a $1/2$ approximation for this problem even with small budgets. Devanur et al.~\cite{DJK13} give a randomized primal dual algorithm that gives a unified analysis of~\cite{KVV90,AGKM11,MSVV07}. We refer the reader to~\cite{Mehta13} for a survey on the online matching literature.

\section{Optimal Algorithm for Binary Ad Exchange Distribution} \label{sec:binary}

In this section, we consider a special case where the highest \adx bid (referred to as \adx reward often) of each query is drawn from a distribution $\dist$ of support size two. We consider the general distribution in Section~\ref{sec:general_alg}. We first provide an algorithm, and later show that this algorithm is optimal. 
Formally we consider the following setting:

\begin{definition}[Binary reward distribution with parameters $q$ and $r$]
We consider the setting where \adx reward distribution $\dist$ is $0$ with probability $q$, and is $r$ with probability $1-q$.
\end{definition} 

Without loss of generality we assume that the two support points are $0$ and $r$, rather than $r_1$ and $r_2$ for $0 < r_1 < r_2$. This is because, in the latter case, we can subtract $r_1$ from each support point, and also from the penalty, and it yields the distribution in the format we need. Also, without loss of generality we assume that the support point $r$ in the distribution is such that $r < c$ where $c$ is the penalty. Note that if $r \geq c$, then clearly whenever the \adx reward is $r$ (i.e., non-zero), an optimal algorithm can always allocate the query to \adx, so there is nothing to study here.

\subsection{An Optimal Algorithm}
\label{subsec:binary-optimal}

Now we propose a simple greedy algorithm (we basically specialize Algorithm~\ref{alg:general} for binary distributions), analyze its performance and establish its optimality. The analysis can be extended to the more general distributions of \adx rewards, but with more involved techniques. We do this in section~\ref{sec:general_alg}. 


Algorithm \ref{alg:binary} is our algorithm for binary reward distributions. Here, we compute an appropriate threshold $s$ as a pre-processing step. At arrival of a query, let $a$ be the available advertiser (i.e., an advertiser with an edge to the incoming vertex) with the lowest satisfaction ratio $\si(a)$. The algorithm allocates the impression to \adx if and only if $\si(a) \geq s$ and the query has non-zero \adx reward of $r$. I.e., the algorithm first greedily allocates queries to available advertisers that are furthest from being satisfied, no matter how large the \adx weight of arriving queries. However, when the advertisers are satisfied to some extent (i.e., their $\si \geq s$), satisfying contracts becomes less of a priority, and \adx is preferred when it offers non-zero reward.

\begin{algorithm}[htb]
\caption{Optimal algorithm for a binary \adx bid distribution}
\label{alg:binary} \small
\textbf{Input:}{ Binary \adx distribution with parameter $q$ and $r$, penalty $c$, and supply factor $f$.}\\
\textbf{Preprocessing:} Set the threshold $s=\max\left(0, 1+fq\ln (1-\frac{r}{c})\right)$ (see Claim \ref{claim:threshold}).\\
 \For{each query arriving online}{
    Let $a$ be a matching advertiser with the lowest satisfaction ratio.\\
    \If{$\si(a)=1$}{
        Assign the impression to \adx.
    }
    \ElseIf{$\si(a) \geq s$ and \adx reward is $r$}{
        Assign the impression to \adx.
      }
    \Else{Assign the impression to $a$.}
    }
\end{algorithm}
Before proving the competitive ratio, we set some notation that we use in our analysis throughout the paper. These concepts are also demonstrated in Figure \ref{fig:binary}.
\begin{figure}[htb]
    \includegraphics[width=14cm,height=6.5cm]{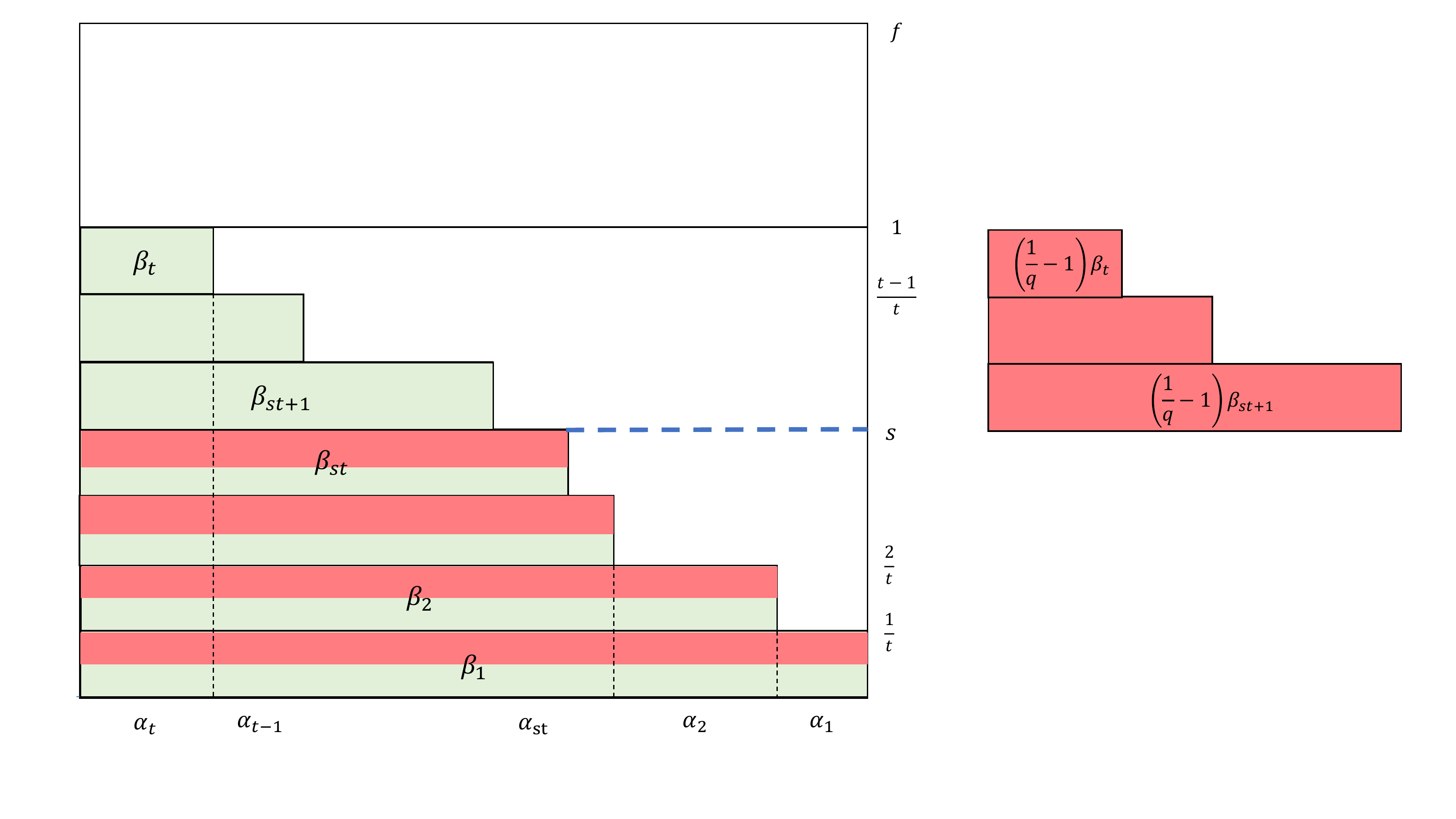}
    \caption{\footnotesize Analysis of Algorithm \ref{alg:binary}: The figure can be viewed from the POV of a single advertiser, as well as from the POV of all advertisers. For a single advertiser $a$, the demand $n_a$ is divided into $t$ intervals. Red colored rectangles represent queries with \adx reward $r$, and green rectangles represent queries with \adx reward 0. The threshold is represented by $s$. The figure on the LHS is for allocation to the contract advertiser $a$, while the figure on the RHS is for allocation to \adx. In the LHS figure, each horizontal rectangle of height $1/t$ represents a value of $\beta_j$. When smallest \si is below the threshold, all the queries (regardless of \adx reward) are allocated to the contract, so each height $1/t$ rectangle below $s$ is constituted by both red and green queries. Above the threshold $s$, all rectangles are only colored green since only queries with \adx reward 0 are allocated to a contract. In the RHS figure, since \adx gets allocated only queries of non-zero value, all rectangles are red. Also, for each query with \adx reward 0 allocated to an advertiser with \si greater than $s$, in expectation $(1/q-1)$ queries get allocated to AdEx (because in expectation for every query with \adx reward $0$, $1/q-1$ queries have reward $r$). The $\alpha$'s in the bottom of the figure come into picture when all the advertisers are taken together. We will soon show that $\alpha_j - t(\beta_j - \beta_{j+1})$. Thus the bottom right corner piece rectangle in the LHS figure represents $\alpha_1/t$ etc.}
    \label{fig:binary}
    \vspace{-0.1cm}
\end{figure} 
Let $t$ be a sufficiently large integer used to discretize the total demand of each advertiser into equal intervals of length $1/t$. The right picture to have in mind is $n_a \gg t \gg 1$. We call any given advertiser $a$ to be of \textit{type} $j$, if at the \emph{end of the algorithm}, $\si(a)\in(\frac{j-1}{t},\frac{j}{t}]$. For type $1$ alone we let the $\si$ interval be closed on both sides, namely $[0,\frac{1}{t}]$. Let $A_j$ be the set\footnote{Note that $A_j$ is a random set depending on the realization of \adx rewards over all queries.} of advertisers of type $j$; and let $\alpha_j=\E[\sum_{a\in A_j}n_a]$ be the total demand of advertisers in $A_j$. For simplicity we assume that an advertiser $a \in A_j$ gets allocated exactly $\frac{j}{t}n_a$ impressions: this lead to an additive error $O(\frac{1}{t})$ in analysis, which is negligible when $t\to\infty$. Let $\beta_j$ be the expected total number (across all advertisers) of allocated impressions s.t., at the time of allocation the assigned advertiser had satisfaction ratio $\in [\frac{j-1}{t},\frac{j}{t})$.
Finally, let $N=\sum_{a\in A} n_a$ be the total demand of all advertisers. 




By definition of $\alpha, \beta$ we get the following (see also Figure \ref{fig:binary}):
\begin{equation}\label{eqn:defofbeta}
    \beta_j=\sum_{a\in\cup_{\ell\geq j}A_{\ell}}\frac{1}{t}n_a=\frac{1}{t}(N-\sum_{\ell<j}\alpha_\ell).
\end{equation}

Thus 
\begin{equation}\label{eqn:alphafrombeta}
    \alpha_j=t(\beta_{j}-\beta_{j+1}).
\end{equation}


\begin{lemma}\label{lem:alphaBeta}
Based on definition of $\alpha, \beta$ as described, for any $j\leq t-1$,
\begin{equation}\label{eqn:casestudybetageneral}
\sum_{\ell\leq j}f\alpha_\ell\leq  \begin{cases} \sum_{\ell\leq st}\beta_\ell+\sum_{st<\ell\leq j}\frac{1}{q}\beta_\ell, &\mbox{if } j\geq st; \\
\sum_{\ell\leq j}\beta_\ell, & \mbox{if } j< st. \end{cases} \end{equation}
\end{lemma}
\begin{proof}



The RHS represents the set of queries that, when they arrived, the most deserving (lowest $\si$) contractual advertiser that was eligible was of type at most $j$. To see this note that when the lowest $\si$ is $\frac{j}{t} < s$, every arriving query is allocated to the contract (hence the second line of RHS). When the lowest $\si$ is at least $s$, only a $q$ fraction of the considered queries are allocated to the contract --- thus the considered queries = allocated queries / q, which is the first line of RHS. 

The LHS represents the number of queries that were allocated to an advertiser of type at most $j$ in the supply-factor-determining-offline-solution.

It is immediate that LHS is at most RHS because every query counted in the LHS will count for RHS when it arrives.
\end{proof}

Notice that the total expected reward of the algorithm can be divided into the following parts: 
\begin{itemize}
    \item The baseline penalty is if no impression is allocated to contracts, the total such penalty is $-Nc$. The total \adx reward that may be obtained by assigning everything to \adx is $Nf(1-q)r$. The next points capture the change to the objective when we move away from this extreme solution of giving everything to \adx.
    \item Any impression that is allocated to an advertiser with satisfaction ratio $\frac{j}{t} < s$ (which is the set of impressions counted in $\beta_j$ for $j \leq st$), with probability $(1-q)$, loses a reward of $r$ from \adx. Thus in expectation each impression has reward $c-(1-q)r$ added to the objective;
    \item Each time an impression is allocated to an advertiser with satisfaction ratio $\frac{j}{t} \geq s$ (which is the set of impressions counted in $\beta_j$ for $j > st$), the impression always has reward $0$ for \adx, but adds $c$ to the objective.
\end{itemize}
Therefore the expected total reward $\alg$ of the algorithm is
\begin{equation}\label{eqn:alg-binary}
\alg=Nf(1-q)r-Nc+\sum_{j\leq st}(c-(1-q)r)\beta_j+\sum_{st < j\leq t}c\beta_j.
\end{equation}
We can add \eqref{eqn:alphafrombeta} and \eqref{eqn:casestudybetageneral} as constraints, to get a linear program that lower bounds the reward of the algorithm as follows:

\begin{align}
    \textrm{minimize}\quad &Nf(1-q)r-Nc+\sum_{j\leq st}(c-(1-q)r)\beta_j+\sum_{st<j\leq t}c\beta_j& \label{eqn:lp-binary}\\
    s.t.    \quad&ft\beta_1-ft\beta_{j+1}\leq \sum_{\ell\leq j}\beta_\ell,&\forall j, 1\leq j\leq st; \nonumber\\
    \quad&ft\beta_1-ft\beta_{j+1}\leq \sum_{\ell\leq st}\beta_\ell+\sum_{st<\ell\leq j}\frac{1}{q}\beta_\ell,&\forall j, st<j\leq t; \nonumber\\
    \quad&\beta_1=\frac{N}{t};& \nonumber\\
    \quad& \beta_j\geq 0,&\forall j, 1\leq j\leq t. \nonumber
\end{align}

The constraints are explained immediately by expanding and doing a telescopic summation using \eqref{eqn:alphafrombeta} and \eqref{eqn:casestudybetageneral}. We set $\beta_1 = N/t$ because in all but pathological instances we have that every advertiser ends up with at least $/1t$ fraction of their demand satisfied (note that $t$ is large, just that $n_a \gg t \gg 1$). Even in the pathological instances where this is not true, i.e., only $\beta_1 < N/t$ holds, by setting $\beta_1 = N/t$, there is just a $O(1/t)$ additive error we have introduced. Namely, when proving optimality of our algorithm, we will just have proved it up to additive $O(1/t)$ terms. From now on, we take $\beta_1 = N/t$.

\begin{claim}\label{claim:binary_LP}  By setting $\beta$ values as follows we get an optimal solution to the linear program \eqref{eqn:lp-binary}:
\[\beta^*_j= \begin{cases} \frac{N}{t}\left(1-\frac{1}{tf}\right)^{j-1}, &\mbox{if } j\leq st+1; \\
\frac{N}{t}\left(1-\frac{1}{tf}\right)^{st}\left(1-\frac{1/q}{tf}\right)^{j-st-1}, & \mbox{if } j>st+1. \end{cases}\]
\end{claim}
\begin{proof}
We prove that there exists an optimal solution of the LP \eqref{eqn:lp-binary} such that all non-trivial constraints are tight. Then the claim follows by observing that $\beta^*$ as defined satisfy this tightness property. To show that the $\beta^*$ leads to tightness, start with a simple assignment of $\beta^*_1=\frac{N}{t}$ and iteratively find the solution to the system of linear equations formed by replacing the inequalities with equalities. This is straightforward. 

To show why tightness is wlog, for any $\beta$ being an optimal solution to the above LP, let $j$ be the smallest index such that the corresponding constraint (of either type) is not tight. If $j\neq st$, we can see that there exists $\epsilon$ such that $\beta_{j+1}\leftarrow\beta_{j+1}-\epsilon$, $\beta_{j+2}\leftarrow\beta_{j+2}+\epsilon$ is a new feasible solution with the objective staying the same, while the $j-$th constraint becomes tight.
Otherwise, if $j=st$, then $\beta_{j+1}\leftarrow\beta_{j+1}-\epsilon$, $\beta_{j+2}\leftarrow\beta_{j+2}+q\epsilon$ is a new feasible solution with the objective decrease by $(c-(1-q)r)\epsilon-cq\epsilon=(1-q)(c-r)\geq 0$, while the $j-$th constraint can become tight. By repeating this process we can construct a solution, with at least the same objective, in which all non-trivial constraints becoming tight.
\end{proof}
We can use the above observations on structure of $\alg$ to compute the appropriate threshold in the following claim:
\begin{claim} \label{claim:threshold}
The objective of the algorithm is maximized when the threshold is set to $s=\max\left(0, 1+fq\ln (1-\frac{r}{c})\right)$.
\end{claim}
\begin{proof}
Using Claim \ref{claim:binary_LP} we have,
\begin{eqnarray}
\alg&\geq&Nf(1-q)r-Nc+\sum_{j\leq st}(c-(1-q)r)\beta^*_j+\sum_{st<j\leq t}c\beta^*_j\nonumber\\
&=&Nf(1-q)r-Nc+\left(c-(1-q)r\right)Nf\left(1-\left(1-\frac{1}{tf}\right)^{st}\right)+cqfN\left(1-\frac{1}{tf}\right)^{st}\left(1-\left(1-\frac{1/q}{tf}\right)^{t-st}\right)\nonumber\\
&=&Nf(1-q)r-Nc+\left(c-(1-q)r\right)Nf(1-e^{-\frac{s}{f}})+cqfNe^{-\frac{s}{f}}(1-e^{-\frac{1-s}{qf}})\nonumber\\
&=&Nc(f-1)+(1-q)(r-c)fNe^{-x}-qfNce^{\frac{1-q}{q}x-\frac{1}{qf}},\nonumber
\end{eqnarray}

where $x=\frac{s}{f}\in[0,\frac{1}{f}]$. Then to maximizes the reward, we consider the following expression in the right hand side:
\begin{equation}\label{eqn:alglb-binary}
RHS(x) = Nc(f-1)+(1-q)(r-c)fNe^{-x}-qfNce^{\frac{1-q}{q}x-\frac{1}{qf}}.
\end{equation}
For optimizing this threshold we take the derivative over $x$, and compare the obtained value with the boundary values for $x$. We have,
\begin{equation*}
    RHS'(x)=(1-q)(c-r)fNe^{-x}-(1-q)fNce^{\frac{1-q}{q}x-\frac{1}{qf}}.
\end{equation*}
We have the unique zero point of $RHS'(x)$ is
$x^*= q(\ln(1-r/c)+\frac{1}{qf})$. Since the allowed range of $x^*$ is $[0,\frac{1}{f}]$, we need to consider the following two cases. When $x^*\in[0,\frac{1}{f}]$, $RHS(x)$ is maximized at $x^*$. Then we set $s^*=fx^*=1+fq\ln(1-r/c)$ in the algorithm, with 
\begin{equation*}
\alg\geq RHS(x^*)=Ncf\left((1-1/f)-(1-r/c) ^{1-q}e^{-1/f}\right).
\end{equation*}
When the solution found is not in range, note that the $x^*$ can only be less than $0$ and never greater than $1/f$. This is because $r < c$, and thus clearly $fx^*=1+fq\ln(1-r/c) < 1$. Given the concavity of the objective, this means that in such a case optimality is achieved at $x=0$. Then we set $s^*=0$ in the algorithm, with
\begin{equation*}
\alg\geq RHS(0)=Ncf\left((1-1/f)+(1-q)(1-r/c)-qe^{-\frac{1}{qf}}\right).
\end{equation*}
\end{proof}

\paragraph{Useful insights.} Interesting insights already flow out of this binary support distribution case. It shows that the optimal threshold $s^*$ that we set is an affine function of the supply factor $f$. Higher the supply factor, lower the threshold we set (note that the coefficient of $f$ in $s^*$, namely $q\ln(1-r/c)$ is negative). Also, the dependence on the penalty $c$ and \adx reward $r$ are quite non-trivial and intriguing. The binary support is often a good first-order approximation of reality when we bucket bids into ``high'' and ``low'' types.

\subsection{Optimality of Algorithm~\ref{alg:binary}}\label{sec:binary-hardness}
We now prove the optimality of the algorithm in the previous section by showing an example for which no algorithm can perform better. Consider a binary distribution with parameter $q$ and $r$ as defined earlier. We use a modification of the ``upper triangular graph'' instance of~\cite{KVV90} as follows:

\begin{restatable}{example}{kvv}\label{ex:kvv-general}
Suppose that there are $m$ advertisers, and each advertiser demands $n$ impressions. There are $fmn=fN$ queries arriving in $m$ groups  $G_1,\cdots,G_m$, with queries in group $G_i$ have an edge to the same $m-i+1$ advertisers determined as follows: consider a random permutation $\pi: [m]\to[m]$, then the queries in group $G_i$ are available to advertisers $j$ with $\pi(j)\geq i$.
\end{restatable}
At a high-level, in this instance, all advertisers are available to the first group of queries arriving. Then with each group one random advertiser is removed from the set of available advertisers to the group. We next argue that Algorithm \ref{alg:binary} is optimal for this instance by showing that any online algorithm will not lead to a better reward.

\begin{theorem}\label{thm:binary-hardness}
For Example~\ref{ex:kvv-general}, the competitive ratio of any randomized online algorithm matches the competitive ratio obtained by Algorithm \ref{alg:binary} up to a small additive factor $O\left(\frac{1}{m}\right)$.
\end{theorem}
\begin{proof}

First we have the following observation about deterministic algorithms.
By Yao's minimax principle, we only need to consider the performance of any deterministic algorithm over the randomness of the instance.

Fix any deterministic algorithm. Let $q_{ij1}$ be the fraction of queries in $G_i$ with \adx reward $0$ that is allocated to advertiser $\pi^{-1}(j)$, and $q_{ij2}$ be the fraction of queries in $G_i$ with \adx reward $r$ that is allocated to advertiser $\pi^{-1}(j)$. Then for $u=1$ and 2,
\begin{equation*}
    \E_\pi[q_{iju}]\leq\begin{cases} \frac{1}{m-i+1}, &\mbox{if } j\geq i; \\
0, & \mbox{if } j<i. \end{cases}
\end{equation*}
Also later we use $\E_\pi[q_{iju}]=\E_\pi[q_{imu}]$.
This is because for each $i$, there are $m-i+1$ random advertisers that have an edge connected to impressions in $G_i$. If $j\geq i$, then $\pi^{-1}(j)$ is a uniformly at random advertiser among this group of $m-i+1$ advertisers. Thus $\E_{\pi}[q_{iju}]\leq \frac{1}{m-i+1}$ and for any $j,j'\geq i$ it holds $\E_{\pi}[q_{iju}]=\E_{\pi}[q_{ij'u}]$. If $j<i$, then advertiser $\pi^{-1}(j)$ does not have an edge to impressions in $G_i$. Then the expected reward  we get from the algorithm, using the same reasoning from the previous section, is

\begin{eqnarray*}
    -Nc+fN(1-q)r+\sum_{i=1}^{m}\sum_{j=i}^{m}\left(\frac{fN}{m}q\E_\pi[q_{ij1}]c+\frac{fN}{m}(1-q)\E_\pi[q_{ij2}](c-r)\right),
\end{eqnarray*}
Here the first term and the second term are the total reward from not allocating anything to the contract advertisers, while the third term is the total reward gain from the allocation of the algorithm: there are in expectation $\frac{fN}{m}q$ queries with \adx reward $0$ (or $\frac{fN}{m}(1-q)$ with reward $r$) from group $G_i$ and $\E_\pi[q_{ij1}]$ (or $\E_\pi[q_{ij2}]$) fraction of them are allocated to advertiser $\pi^{-1}(j)$, with each impression contributing to a reward gain $c$ (or $c-r$) compared to being allocated to \adx.

As we discussed $\E_{\pi}[q_{iju}]=\E_{\pi}[q_{imu}]$ for any $j\geq i, u=1,2$. Hence we can simplify the overall expectation for all $j \geq i$:

\begin{eqnarray*}
-cN+fN(1-q)r+\sum_{i=1}^{m}(m-i+1)\left(\frac{fN}{m}q\E_\pi[q_{im1}]c+\frac{fN}{m}(1-q)\E_\pi[q_{im2}](c-r)\right).
\end{eqnarray*}
Then the reward of the algorithm is upper bounded by the solution of the following linear program, where $y_{iu}$ variables represent the expected value $\E_{\pi}[q_{imu}]$. 
\begin{gather}\textrm{maximize}\quad \displaystyle-cN+fN(1-q)r+\frac{fN}{m}\sum_{i=1}^{m}(m-i+1)\left(qy_{i1}c+(1-q)y_{i2}(c-r)\right)\nonumber\\
\begin{array}{rrlll}\label{lp:negative-binary}
    s.t.    \quad&\displaystyle\sum_{i=1}^{m}\left(\frac{fN}{m}qy_{i1}+\frac{fN}{m}(1-q)y_{i2}\right)&\leq& \displaystyle\frac{N}{m};&\\
    \quad&0\leq y_{i1},y_{i2}&\leq&\displaystyle\frac{1}{m-i+1},&\forall i, 1\leq i\leq m.
\end{array}
\end{gather}
Here the left hand side of the first constraint is the total expected number of allocated impressions to advertiser $\pi^{-1}(m)$, which is at most $n=\frac{N}{m}$. 

Next, we show a structure on any optimal solution to this LP, that captures a threshold based behavior that we can be related to the algorithm we presented in the previous section:
\begin{lemma}\label{lem:negativelp_thresholds_binary}
For an optimal solution $\mathbf{y}$ to the above LP, there exists thresholds $1 \leq z_2\leq z_1\leq m$, such that, $y_{iu}=\frac{1}{m-i+1}$ for $i<z_u$, and $y_{iu}=0$ for $i>z_u$ for $u=1,2$. 
\end{lemma}
\begin{proof}
First, we show that in any optimal solution $\mathbf{y}$ and a threshold $z_1$, such that $y_{i1}=\frac{1}{m-i+1}$ for $i<z_1$, and $y_{i1}=0$ for $i>z_1$. Then a similar claim follows for $z_2$.
To show such a threshold behavior holds for $y_{i1}$ values in any optimal solution $\mathbf{y}$, where $i \leq z_1$, we equivalently argue that there cannot be $i < i'$ such that $0<y_{i1}<\frac{1}{m-i+1}$, $0<y_{i'1}<\frac{1}{m-i'+1}$ for $i<i'$.
Let us assume by contradiction that such $i,i'$ exists.
Then setting $y_{i1}\leftarrow y_{i1}+\epsilon$, $y_{i'1}\leftarrow y_{i'1}-\epsilon$ for small enough $\epsilon$ leads to a new feasible solution since all constraints are still feasible. Furthermore, in the objective function $y_{i1}$ has coefficient $(m-i+1)\frac{fN}{m}qc>(m-i'+1)\frac{fN}{m}qc$, which is the coefficient of $y_{i'1}$. Thus after perturbing $\mathbf{y}$ this way we get feasible solution with a larger objective value. This contradicts the assumption of $\mathbf{y}$ being optimal.

Next, we show that $z_2 \leq z_1$, i.e.~the thresholds are monotone.
For any optimal solution $\mathbf{y}$, if $z_1< z_2$, then for $i=z_{2}$, $y_{i1}<\frac{1}{m-i+1}$, while $y_{i2}=\frac{1}{m-i+1}$. Then setting $y_{i1}\leftarrow y_{i1}+\frac{1}{q}\epsilon$, $y_{i2}\leftarrow y_{i2}-\frac{1}{1-q}\epsilon$ for small enough $\epsilon$ leads to a new feasible solution since all constraints are still feasible. Furthermore, the increase of the objective due to $y_{i1}$ is $(m-i+1)\frac{fN}{m}c\epsilon>(m-i+1)\frac{fN}{m}(c-r)\epsilon$ which is the decrease of the objective due to $y_{i2}$. Thus after perturbing $\mathbf{y}$ this way we get a new feasible solution with a larger objective value. This contradicts the assumption of $\mathbf{y}$ being optimal.
\end{proof}

From the above two lemmas, we know that the optimal strategy for Example~\ref{ex:kvv-general} has the following form: for queries in group $G_1,\cdots,G_{z_2}$, all impressions are allocated uniformly to all available advertisers; for queries in group $G_{z_2+1},\cdots,G_{z_1}$, only queries with \adx reward $0$ are allocated uniformly to all available advertisers; for queries in group $G_{z_{1}+1},G_{z_{1}+2},\cdots,G_{m}$, no impression is allocated a contract. 


By setting the $y$ values, as determined by Lemma~\ref{lem:negativelp_thresholds_binary}, we can simplify the objective function of linear program \eqref{lp:negative-binary} with threshold $z_1$ and $z_2$ and bound the reward $\alg$ obtained from an online algorithm as follows: The objective can be written as 
\[-cN+(1-q)r+ \left(\sum_{i=1}^{z_2}\left(qc+(1-q)(c-r)\right)+\sum_{i=z_2+1}^{z_1}qc\right) = -cN+(1-q)r +z_1qc + z_2 (1-q)(c-r) z_2.\]

Then we get,
\begin{gather}\alg\leq\max_{z_1,z_2}\quad \displaystyle -cN+(1-q)fNr+z_1qc + z_2 (1-q)(c-r) z_2\label{lp:negative-binary-simplified}\\
\begin{array}{rrlll}\nonumber
    s.t.    \quad&\displaystyle\sum_{i=1}^{z_2} f\cdot\frac{1}{m-i+1}+\sum_{i=z_2+1}^{z_1}f\cdot\frac{1}{m-i+1}q&=& \displaystyle 1&
\end{array}
\end{gather}
When $m$ is large enough, the constraint can be replaced by
\begin{equation*}
    f\ln\frac{m}{m-z_2}+fq\ln\frac{m-z_2}{m-z_1}=1.
\end{equation*}
Let $x=\frac{m}{m-z_2}$, then $x\in[0,\frac{1}{f}]$. Then we can express $z_1$ and $z_2$ by $x$ as follows: $z_1=m(1-e^{\frac{x(1-q)}{q}-\frac{1}{fq}})$, and $z_2=m(1-e^{-x})$. Apply these to \eqref{lp:negative-binary-simplified} we have
\begin{eqnarray*}
    \alg&\leq& \max_{x\in [0,\frac{1}{f}]} -cN+(1-q)fNr+m(1-e^{\frac{x(1-q)}{q}-\frac{1}{fq}})qc + m(1-e^{-x})(1-q)(c-r)\\
    &=&\max_{x\in [0,\frac{1}{f}]}Nc(f-1)+(1-q)(r-c)fNe^{-x}-qfNce^{\frac{1-q}{q}x-\frac{1}{qf}}.
\end{eqnarray*}
Notice that the optimization problem here is identical to the optimization problem \eqref{eqn:alglb-binary} that we described in the analysis of Algorithm~\ref{alg:binary}. Thus the upper bound of the performance of any online algorithm for this instance matches the lower bound of the performance of Algorithm~\ref{alg:binary} for any underlying graph. As the optimal offline allocation has the same expected reward for any instance (see Theorem~\ref{thm:opt-general} for a more detailed discussion), we prove the optimality of Algorithm~\ref{alg:binary}.

\end{proof}

\paragraph{Going from Section~\ref{sec:binary} to Section~\ref{sec:general_alg}.} In Section~\ref{sec:general_alg}, we use a similar max-min approach as in Section~\ref{sec:binary}. However, the max-min problem of the algorithm is no-more the simple single-variable concave maximization problem. It is a multi-variate, non-linear and non-convex optimization problem. While we cannot solve it precisely optimally in general, we show a dynamic program that can solve it to almost optimality with a small additive error. Also, while establishing tightness, the task was simpler in Section~\ref{sec:binary} because we had to compare the upper bound from the hard example to the single variable expression and show that these are the same expressions. But in section~\ref{sec:general_alg} we establish that the non-linear mathematical programs obtained in the maximization problem of the algorithm and in the hard example are identical. The non-trivial roles that $f$, the \adx distribution, and the penalty $c$ play in determining the optimal thresholds is the core contribution of our work.

\section{Optimal Algorithm for a General Ad Exchange Distribution} \label{sec:general_alg}
 In this section we consider a general \adx reward distribution. More formally, we have a constant penalty $c$ and each query has an \adx reward drawn from a discrete distribution $\dist$ with a fixed support size $d$\footnote{The assumption on a fixed support, can be relaxed using a standard discretization approach at a small cost in the competitive ratio that depends on this discretization.}, and the supply factor is $f$. 
We propose a threshold-based algorithm in which a set of thresholds $s_1,...,s_d$ are chosen based on an optimization problem that takes $\dist,f,c$ into account. 
We then show that this algorithm is optimal. We consider the same instance used in Section \ref{sec:binary-hardness}, and show that the optimal solutions on this instance for the two optimization problems are the same when the number of advertisers is sufficiently large.  
 Finally, we show that the binary distribution is the worst-case distribution for any class of algorithm with a fixed mean $\mu$. This allows us to obtain a competitive ratio, that depends on $\mu,c,f$ using our results in Section \ref{sec:binary}.


\subsection{Optimal Algorithm for General AdEx Distribution}

In this section, we provide a threshold-based algorithm 
, and in future sections we discuss the computational aspects and prove tightness. First, let us formalize the notation:

\begin{definition}[\adx distribution with parameters $\left( (r_i,q_i)_{i \in [d]}\right)$] We consider an \adx distribution $\dist$ with support size $d$, rewards $0=r_1 \leq r_2 \leq ...\leq r_d$, where probability of that the reward is $r \leq r_i$ is $q_i$. Also we set $q_0=0, q_d=1$. 
\end{definition}

In other words, for $r\sim \dist$, with probability $q_{u}-q_{u-1}$, we have $r=r_{u}$, $\forall 1\leq u\leq d$; $q_d=1$.
Without loss of generality, we assume $r_1=0$. Otherwise, we can shift the rewards and the penalty by $-r_1$, since $r_1$ is the smallest reward from any allocation. We also assume $r_d\leq c$, since otherwise, when a query with \adx reward at least $c$ arrives, an optimal strategy always allocates the impression to \adx, and hence we can disregard such queries.

Our algorithm is presented in Algorithm~\ref{alg:general} (see Section~\ref{sec:intro}). For any query that arrives, if $a$ is the advertiser the lowest satisfaction ratio, and $\si(a) \in[s_{u-1},s_{u})$, then the impression is allocated to $a$ if and only if its \adx reward $r\leq r_{d+1-u}$. Here we define $s_0=0$ for completeness.

We use the same setup as we in analysis of the algorithm in Section \ref{sec:binary}. Recall that we discretize the algorithm into $t$ steps.
An advertiser $a$ has type $j$ if at the end of the algorithm, $\si(a)\in[\frac{j-1}{t},\frac{j}{t})$. We defined $\alpha_j=\E[\sum_{a\in A_j}n_a]$ be the total demand of advertisers in the set $A_j$ of all advertisers of type $j$, and $\beta_j$ be the expected total number of impressions that get allocated to an advertiser $a$ with $\si(a)\in[\frac{j-1}{t},\frac{j}{t}]$ by the algorithm at the time the query arrives. We can relate the values of $\alpha$ and $\beta$ using a similar reasoning as in Lemma \ref{lem:alphaBeta}. Formally,
\vspace{-0.2cm}
\begin{lemma}\label{lem:general_alphabeta}
Consider an \adx distribution with parameters $\left( (r_i,q_i)_{i \in [d]}\right)$, where penalty $c$, $r_d \leq c$, and let $\alpha, \beta$ be as defined above. We have,
\vspace{-0.2cm}
\begin{equation}\label{eqn:general_alg}
\sum_{\ell\leq j}f\alpha_\ell\leq  \begin{cases}
\frac{1}{q_d}\sum_{0<\ell\leq j}\beta_\ell,&\mbox{if } j\leq s_1 t;\\
\frac{1}{q_d}\sum_{0<\ell\leq s_1 t}\beta_\ell+\frac{1}{q_{d-1}}\sum_{s_1 t<\ell\leq j}\beta_\ell,&\mbox{if } s_1 t<j\leq s_2 t;\\
\cdots & \cdots\\
\sum_{u=1}^{d-1}\sum_{s_{u-1} t<\ell\leq s_u t}\frac{1}{q_{d+1-u}}\beta_\ell+\sum_{s_{d-1} t<\ell\leq j}\frac{1}{q_1}\beta_\ell, &\mbox{if } s_{d-1} t<j\leq s_d t=t. \end{cases} \end{equation}
\end{lemma}
The proof is omitted, since it is a straightforward extension of Lemma \ref{lem:alphaBeta} that was used for the binary distribution.

A similar case by case analysis as in \eqref{eqn:alg-binary}, allows us to write an expression for the total expected reward by considering the following parts: 
\begin{itemize}
    \item The baseline penalty is if no impression is allocated to contracts, the total penalty is $-Nc$.
    \item The total \adx reward that may be obtained is $\sum_{u=1}^{d}fN(q_u-q_{u-1})r_u$.
    \item Any impression that is allocated to an advertiser with satisfaction ratio in $(s_{u-1},s_{u}]$, in expectation gets a reward of $\E_{\dist}[r|r\leq r_{d+1-u}]$. Thus in expectation each query has reward $c-\E_{\dist}[r|r\leq r_{d+1-u}]$ added to the total penalty.
\end{itemize}
Therefore the expected total reward $\alg$ of the algorithm is
\begin{equation*}
    \alg=-cN+\sum_{u=1}^{d}fN(q_u-q_{u-1})r_u+\sum_{u=1}^{d}\sum_{j=s_{u-1} t+1}^{s_u t}N\beta_j(c-\E_{\dist}[r|r\leq r_{d+1-u}]).
\end{equation*}
We can add \eqref{eqn:defofbeta}, \eqref{eqn:alphafrombeta} and Lemma \ref{lem:general_alphabeta} to constraints of a linear program to lower bound the reward of the algorithm as follows:

{\small
\begin{align} 
    \textrm{minimize}\quad \alg&\label{lp:general}\\
    s.t.    \quad&ft\beta_1-ft\beta_{j+1}\leq \sum_{\ell\leq j}\frac{1}{q_d}\beta_\ell,&\forall j\leq s_1 t;\nonumber\\
    \quad&ft\beta_1-ft\beta_{j+1}\leq \sum_{\ell\leq s_1 t}\frac{1}{q_d}\beta_\ell+\sum_{s_1 t<\ell\leq j }\frac{1}{q_{d-1}}\beta_\ell,&\forall s_1 t<j\leq s_2 t;\nonumber\\
    &\cdots & \cdots\nonumber\\
    \quad&ft\beta_1-ft\beta_{j+1}\leq \sum_{\ell\leq s_1 t}\frac{1}{q_d}\beta_\ell+\sum_{s_1 t<\ell\leq s_2 t}\frac{1}{q_{d-1}}\beta_\ell+\cdots&\nonumber\\
    &\quad\quad\quad\quad\quad\quad+\sum_{s_{d-2} t<\ell\leq s_{d-1} t}\frac{1}{q_2}\beta_\ell+\sum_{s_{d-1} t<\ell\leq j}\frac{1}{q_1}\beta_\ell,&\forall s_{d-1} t<j\leq s_d t=t;\nonumber\\
    \quad&\beta_1=\frac{N}{t};&\nonumber\\
    \quad& \beta_j\geq 0,&\forall j\leq t.\nonumber
\end{align}
}
Next, using similar arguments as in Claim \ref{claim:binary_LP} we argue that by solving a system of linear equations formed by the LP constraints, we can obtain optimal solutions. For this consider the following $\beta$ values:
{\small
\begin{equation}\label{eqn:beta}
\beta^*_j=\begin{cases} \frac{N}{t}\left(1-\frac{1/q_d}{tf}\right)^{j-1}, &\mbox{if } j\leq s_1 t+1; \\
\frac{N}{t}\left(1-\frac{1/q_d}{tf}\right)^{s_1 t-s_0 t}\left(1-\frac{1/q_{d-1}}{tf}\right)^{j-s_1t-1}, & \mbox{if } s_1t+1<j\leq s_2t+1; \\
\cdots\\
\frac{N}{t}\left(1-\frac{1/q_d}{tf}\right)^{s_1 t-s_0 t}\cdots\left(1-\frac{1/q_{d+2-u}}{tf}\right)^{s_{u-1}t-s_{u-2}t}\left(1-\frac{1/q_{d+1-u}}{tf}\right)^{j-s_{u-1}t-1}, & \mbox{if } s_{u-1}t+1<j\leq s_ut+1.\\
\cdots\\
\frac{N}{t}\left(1-\frac{1/q_d}{tf}\right)^{s_1t-s_0t}\cdots\left(1-\frac{1/q_{d+2-u}}{tf}\right)^{s_{d-1}t-s_{d-2}t}\left(1-\frac{1/q_1}{tf}\right)^{j-s_{d-1}t-1}, & \mbox{if } s_{d-1}t+1<j\leq s_dt=t.
\end{cases}
\end{equation}
}
\begin{claim}\label{claim:beta_general}
The $\beta^*_j, 1 \leq j \leq t$ values defined above, form an optimal solution to LP \eqref{lp:general}.
\end{claim}
The argument is similar to the proof of Claim \ref{claim:binary_LP}, and a sketch is provided in Appendix \ref{app:general}. Next, similarly to Section~\ref{sec:binary}, the performance of the algorithm is lower bounded by the following formula based on LP \eqref{lp:general}: 
\begin{equation}\label{eqn:generallbopt}
ALG(s_1,...,s_d) \geq-cN+\sum_{u=1}^{d}fN(q_u-q_{u-1})r_u+\sum_{u=1}^{d}\sum_{j=s_{u-1}t+1}^{s_ut}\beta^*_j(c-\E_{F}[r|r\leq r_{d+1-u}]).
\end{equation}
For the convenience of future reference, we define the following optimization problem for an arbitrary instance of the problem when we have a fixed penalty $c$, and \adx distribution with support size $d$, $m$ advertisers, and $N$ total demand:
\begin{optimization}[Maximization Problem] \label{def:lb}
Given an \adx distribution $\dist$ with parameters $(r_i,q_i)_{i \in [d]}$, find $0\leq s_1\leq s_2\leq \cdots\leq s_d=1$ that maximizes the following objective such that $\beta^*_j$ values satisfy the above constraints:
\begin{equation*}
    \lb_{m,N}(s_1,\cdots,s_d):=-cN+\sum_{u=1}^{d}fN(q_u-q_{u-1})r_u+\sum_{u=1}^{d}\sum_{j=s_{u-1}t+1}^{s_ut}\beta^*_j(c-\E_{\dist}[r|r\leq r_{d+1-u}]).
\end{equation*}
\end{optimization}
In the next section, we show that Algorithm \ref{alg:general} is also optimal:
\begin{theorem}\label{thm:general_positive}
For any $f \geq 1$, and \adx distribution with parameters $(r_i,q_i)_{i \in [d]}$, Algorithm \ref{alg:general} with thresholds determined by Optimization Problem \ref{def:lb} leads to an optimal algorithm.
\end{theorem}

But before proving the optimality, we describe how we can computationally estimate the thresholds, if Optimization Problem \ref{def:lb} is not easy to solve directly. 

\subsection{Computing the Thresholds}\label{sec:DP}
While the threshold is easily computed for the binary distribution setting, for an arbitrary distribution, the optimization problem may not necessarily have computationally efficient solutions. Hence we use dynamic programming to generalize our results to any distribution and use a polynomial-time algorithm at the cost of a small additional error. For this, fix a parameter $0 <\epsilon <1$, and divide the interval $[0,1]$ to multiple of $\epsilon$, we have $1/\epsilon$ buckets. We set the thresholds $s_1,...,s_d$ to be the closest multiple of $\epsilon$ (by rounding down).
We then use a standard dynamic-programming approach that finds the best threshold among the multiples of $\epsilon$. The proof is deferred to Appendix \ref{app:DP}. 
\begin{theorem}\label{thm:DP}
There exists an algorithm with $O(m^3 d^2)$ running time that outputs a feasible set of thresholds $(\hat{s}_1,\hat{s}_2,\cdots,\hat{s}_d)$ such that $$\lb_{m,N}(\hat{s}_1,\cdots,\hat{s}_d)\geq \max_{s_1,\cdots,s_d}\lb_{m,N}(s_1,\cdots,s_d)-O(cN/m).$$
\end{theorem}

\subsection{Tightness for General Reward Distribution}\label{sec:hardness}
\label{sec:general-hardness}

Next, we are going to analyze the performance of the algorithm in the previous section by showing that on the instance we also used for binary distribution, no online algorithm can perform better.
Let $\dist$ be the \adx distribution with parameters $\left( (r_i,q_i)\right)_{i \in [d]}$. Recall the following example: 

\kvv*

We prove in the following theorem that in this instance, no online algorithm can get a reward more than the objective of Optimization Problem~\ref{def:negative_opt}. Since by Theorem~\ref{thm:opt-general} all instances with the same total demand $N$ have identical optimal reward $\opt$,  Theorem~\ref{thm:general_positive} follows immediately from Theorem~\ref{thm:general-hardness}.

\begin{theorem}\label{thm:general-hardness}
For Example~\ref{ex:kvv-general}, the expected optimal reward of any randomized online algorithm is upper bounded by the solution of Optimization Problem~\ref{def:lb}, up to a negligible error for large enough $m$ and $N$.
\end{theorem}

\begin{proof}

First, we have the following observation about deterministic algorithms.
By Yao's min-max principle, we only need to consider the performance of any deterministic algorithm over the randomness of the instance.

Fix any deterministic algorithm. Let $q_{iju}$ be the fraction of queries in $G_i$ with \adx reward $r_u$ that is allocated to advertiser $\pi^{-1}(j)$. Then
\begin{equation*}
    \E_\pi[q_{iju}]\leq\begin{cases} \frac{1}{m-i+1}, &\mbox{if } j\geq i; \\
0, & \mbox{if } j<i. \end{cases}
\end{equation*}
This is due to the following observation: there are $m-i+1$ advertisers incident to vertices in $G_i$. If $j\geq i$, then $\pi^{-1}(j)$ is a random advertiser that is incident to vertices in $G_i$, thus $\E_{\pi}[q_{iju}]=\E_{\pi}[q_{ij'u}]\leq \frac{1}{m-i+1}$ for any $j,j'\geq i$. If $j<i$ , then advertiser $\pi^{-1}(j)$ is not available for queries in $G_i$, thus get zero allocation. Then the expected reward we get from the algorithm is
\begin{eqnarray*}
    -cN+\sum_{u=1}^{d}fN(q_u-q_{u-1})r_u+\sum_{i=1}^{m}\sum_{j=i}^{m}\sum_{u=1}^{d}\frac{fN}{m}(q_u-q_{u-1})\E_\pi[q_{iju}](c-r_u),
\end{eqnarray*}
Here the first term and second terms are the reward from not allocating anything to the contract advertisers, while the third term is the total reward gain: there are in expectation $\frac{fN}{m}(q_u-q_{u-1})$ queries with \adx reward $r_u$ from group $G_i$; $\E_\pi[q_{iju}]$ fraction of them are allocated to advertiser $\pi^{-1}(j)$, with each query having a reward gain of $c-r_u$ compared to the impression being allocated to \adx. Since $\E_{\pi}[q_{iju}]=\E_{\pi}[q_{ij'u}]$ for any $j,j'\geq i$, we can simplify the overall expectation for all $j \geq i$:
\begin{eqnarray*}
-cN+\sum_{u=1}^{d}fN(q_u-q_{u-1})r_u+\sum_{i=1}^{m}(m-i+1)\sum_{u=1}^{d}\frac{fN}{m}(q_u-q_{u-1})\E_\pi[q_{imu}](c-r_u).
\end{eqnarray*}
Then the reward of the algorithm is upper bounded by the solution of the following linear program:
\begin{gather}\textrm{maximize}\quad \displaystyle-cN+\sum_{u=1}^{d}fN(q_u-q_{u-1})r_u+\sum_{i=1}^{m}(m-i+1)\sum_{u=1}^{d}\frac{fN}{m}(q_u-q_{u-1})y_{iu}(c-r_u)\nonumber\\
\begin{array}{rrlll}\label{lp:negative}
    s.t.    \quad&\displaystyle\sum_{i=1}^{m}\frac{fN}{m}\sum_{u=1}^{d}(q_u-q_{u-1})y_{iu}&\leq& \displaystyle\frac{N}{m};&\\
    \quad&0\leq y_{iu}&\leq&\displaystyle\frac{1}{m-i+1},&\forall 1\leq i\leq m;\\
    \quad& y_{iu}&\geq&0,&\forall 1\leq i\leq m.
\end{array}
\end{gather}
The first constraint follows from the fact the total number of allocated impressions $\sum_{i=1}^{m}\frac{fN}{m}\sum_{u=1}^{d}(q_u-q_{u-1})\E_\pi[q_{imu}]$ is at most $n=\frac{N}{m}$ for advertiser $\pi^{-1}(m)$ in expectation.
In the following we characterize any optimal solution to the LP based on a set of threshold values. The argument is similar to proof of Claim \ref{lem:negativelp_thresholds_binary}, hence we defer the proof to appendix.
\begin{lemma}\label{lem:negativelp_thresholds}
For an optimal solution $\mathbf{y}$ to the above LP, there exists thresholds $z_u\in[m]$ for any $u\in[d]$,  such that $y_{iu}=\frac{1}{m-i+1}$ for $i<z_u$, and $y_{iu}=0$ for $i>z_u$. Moreover, For any $u<u'$, we have $z_u\geq z_{u'}$ for this threshold vector $\mathbf{z}$.
\end{lemma}

From the above lemma, we know that the optimal strategy for Example~\ref{ex:kvv-general} has the following form: for queries in group $G_1,\cdots,G_{z_d}$, all impressions are allocated uniformly to all available advertisers; for queries in group $G_{z_d+1},\cdots,G_{z_{d-1}}$, only queries with \adx reward $\leq r_{d-1}$ are allocated uniformly to all available advertisers. For queries in group $G_{z_{j}+1},G_{z_{j}+2},\cdots,G_{z_{j-1}}$, only queries with \adx reward $\leq r_{j-1}$ are allocated uniformly to all available advertisers, $\forall 2\leq j\leq d$. By applying Lemma~\ref{lem:negativelp_thresholds} to the objective function of LP~\eqref{lp:negative}, by setting $y_{iu}=\frac{1}{m-i+1}$ for $i\leq z_u$ and $y_{iu}=0$ for $i>z_u$, we can simplify the objective function as follows:
\begin{equation}\label{eqn:simplified-obj-z}
    -cN+\sum_{u=1}^{d}fN(q_u-q_{u-1})r_u+\frac{fN}{m}\sum_{u=1}^{d}z_u(q_u-q_{u-1})(c-r_u).
\end{equation}

For any $u\in [d]$, let $s'_{u}$ be the satisfaction ratio of each remaining advertiser after the queries of the first $z_{d+1-u}$ groups have arrived.
Observe that after queries $G_1$ have arrived, each advertiser is allocated $\frac{fn}{m}$ impressions, thus $\frac{f}{m}$ fraction of demand of each of the $m$ advertisers is satisfied. After queries in the next group have arrived, ${f}{m-1}$ additional fraction of the demand of each remaining advertiser is satisfied. Using similar arguments we have
\begin{equation*}
    s'_1=\frac{f}{m}+\frac{f}{m-1}+\cdots+\frac{f}{m-z_{d}+1}\approx f\ln\frac{m}{m-z_d}.
\end{equation*}
Here the equation is accurate up to a small $O(\frac{1}{m})$ error, thus is negligible for large enough $m$. After queries in $G_{z_d+1}$ have arrived, each advertiser is allocated $\frac{fnq_{d-1}}{m-z_{d}}$ impressions in expectation, since only queries with \adx reward at most $r_{d-1}$ are allocated to the $m-z_{d}$ remaining advertisers uniformly. Thus after queries in $G_{z_d+1}$ have arrived the satisfaction ratio of each available advertiser increases by $\frac{fq_{d-1}}{m-z_{d}}$. Using similar arguments to group $G_{z_d+2},\cdots,G_{z_{d-1}}$ we have
\begin{equation*}
    s'_2=s'_1+\frac{fq_{d-1}}{m-z_{d}}+\frac{fq_{d-1}}{m-z_{d}-1}+\cdots+\frac{fq_{d-1}}{m-z_{d-1}+}=s'_1+fq_{d-1}\ln\frac{m-z_d}{m-z_{d-1}}=f\ln\frac{m}{m-z_{d}}+fq_{d-1}\ln\frac{m-z_d}{m-z_{d-1}}.
\end{equation*}
Using the same analysis we can get 
\begin{equation}\label{eqn:express-sbyz}
    s'_{u}=\sum_{j=1}^{u}fq_{d+1-j}\ln\frac{m-z_{d+2-u}}{m-z_{d+1-u}}
\end{equation}
for every $u=1,2,\cdots,d$ if we define $q_d=1$ and $z_{d+1}=0$ for completeness. We can express $\textbf{z}$ by $\textbf{s}'$ as follows:
\begin{eqnarray}
z_u&=&m-m\exp\left(-\frac{(s'_1-s'_0)m}{fq_d}-\frac{(s'_2-s'_1)m}{fq_{d-1}}-\cdots-\frac{(s'_{d+1-u}-s'_{d-u})m}{fq_u}\right)\nonumber\\
&=&m\left(1-\exp\left(-\sum_{j=1}^{d+1-u}\frac{(s'_j-s'_{j-1})}{fq_{d+1-j}}\right)\right).\label{eqn:express-zbys} 
\end{eqnarray}

By replacing the values in \eqref{eqn:express-zbys} to the objective function in \eqref{eqn:simplified-obj-z}, we can \textit{upper bound} the reward of \textit{any} online algorithm on Example~\ref{ex:kvv-general} 
(for $m$, and $N$) 
as defined by the following optimization problem: 

\begin{optimization}[Reward of Example \ref{ex:kvv-general}]\label{def:negative_opt}
Consider an \adx distribution with parameters $((r_i,q_i)_{i \in [d]}$. Find thresholds $\textbf{s}'$ that maximize:
\begin{equation*}
    \ub_{m,N}(\textrm{s}')\equiv-cN+\sum_{u=1}^{d}fN(q_u-q_{u-1})r_u+fN\sum_{u=1}^{d}\left(1-\exp\left(-\sum_{j=1}^{d+1-u}\frac{(s'_j-s'_{j-1})}{fq_{d+1-j}}\right)\right)(q_u-q_{u-1})(c-r_u).
\end{equation*}
\end{optimization}

Our goal is to relate such an optimization problem over variables $s'$ with Optimization Problem~\ref{def:lb}, which we used in Section \ref{sec:general_alg} to get a lower bound of the objective of Algorithm~\ref{alg:general}. To show that Optimization Problem~\ref{def:lb} and Optimization Problem~\ref{def:negative_opt} have the same optimal objective when $m$ and $N$ are sufficiently large, it suffices to show the following claim.

\begin{claim}\label{clm:lbub}
    For any $\textrm{s}$ such that $0\leq s_1\leq \cdots s_{d}=1$ and large enough $m,N$, 
    \begin{equation*}
        \lb_{m,N}(s)=\ub_{m,N}(s).
    \end{equation*}
\end{claim}

It follows from Claim~\ref{clm:lbub}, that $\max_{s}\lb_{m,N}(s)=\max_{s}\ub_{m,N}(s)$ for sufficiently large $m$ and $N$. Thus the reward of \textit{any} algorithm on the instance of Example \ref{ex:kvv-general}, captured by the function $\ub$ matches the reward of Algorithm \ref{alg:general}, captured by the function $\lb$, concluding the proof of Theorem~\ref{thm:general-hardness}.
\end{proof}

\bibliographystyle{plainnat}
\bibliography{reference}

\begin{thebibliography}{20}
\providecommand{\natexlab}[1]{#1}
\providecommand{\url}[1]{\texttt{#1}}
\expandafter\ifx\csname urlstyle\endcsname\relax
  \providecommand{\doi}[1]{doi: #1}\else
  \providecommand{\doi}{doi: \begingroup \urlstyle{rm}\Url}\fi

\bibitem[Aggarwal et~al.(2011)Aggarwal, Goel, Karande, and Mehta]{AGKM11}
Gagan Aggarwal, Gagan Goel, Chinmay Karande, and Aranyak Mehta.
\newblock Online vertex-weighted bipartite matching and single-bid budgeted
  allocations.
\newblock In Dana Randall, editor, \emph{Proceedings of the Twenty-Second
  Annual {ACM-SIAM} Symposium on Discrete Algorithms, {SODA} 2011, San
  Francisco, California, USA, January 23-25, 2011}, pages 1253--1264. {SIAM},
  2011.

\bibitem[Agrawal and Devanur(2015)]{AD15}
Shipra Agrawal and Nikhil~R. Devanur.
\newblock Fast algorithms for online stochastic convex programming.
\newblock In Piotr Indyk, editor, \emph{Proceedings of the Twenty-Sixth Annual
  {ACM-SIAM} Symposium on Discrete Algorithms, {SODA} 2015, San Diego, CA, USA,
  January 4-6, 2015}, pages 1405--1424. {SIAM}, 2015.

\bibitem[Agrawal et~al.(2014)Agrawal, Wang, and Ye]{AWY14}
Shipra Agrawal, Zizhuo Wang, and Yinyu Ye.
\newblock A dynamic near-optimal algorithm for online linear programming.
\newblock \emph{Oper. Res.}, 62\penalty0 (4):\penalty0 876--890, 2014.

\bibitem[Balseiro et~al.(2014)Balseiro, Feldman, Mirrokni, and
  Muthukrishnan]{BFMM14}
Santiago~R. Balseiro, Jon Feldman, Vahab~S. Mirrokni, and S.~Muthukrishnan.
\newblock Yield optimization of display advertising with ad exchange.
\newblock \emph{Manag. Sci.}, 60\penalty0 (12):\penalty0 2886--2907, 2014.

\bibitem[Balseiro et~al.(2020)Balseiro, Lu, and Mirrokni]{BLM20}
Santiago~R. Balseiro, Haihao Lu, and Vahab~S. Mirrokni.
\newblock Dual mirror descent for online allocation problems.
\newblock In \emph{Proceedings of the 37th International Conference on Machine
  Learning, {ICML} 2020, 13-18 July 2020, Virtual Event}, volume 119 of
  \emph{Proceedings of Machine Learning Research}, pages 613--628. {PMLR},
  2020.

\bibitem[Choi et~al.(2020)Choi, Mela, Balseiro, and Leary]{CMBL20}
Hana Choi, Carl~F. Mela, Santiago~R. Balseiro, and Adam Leary.
\newblock Online display advertising markets: {A} literature review and future
  directions.
\newblock \emph{Inf. Syst. Res.}, 31\penalty0 (2):\penalty0 556--575, 2020.

\bibitem[Devanur and Hayes(2009)]{DH09}
Nikhil~R. Devanur and Thomas~P. Hayes.
\newblock The adwords problem: online keyword matching with budgeted bidders
  under random permutations.
\newblock In John Chuang, Lance Fortnow, and Pearl Pu, editors,
  \emph{Proceedings 10th {ACM} Conference on Electronic Commerce (EC-2009),
  Stanford, California, USA, July 6--10, 2009}, pages 71--78. {ACM}, 2009.

\bibitem[Devanur and Jain(2012)]{DJ12}
Nikhil~R Devanur and Kamal Jain.
\newblock Online matching with concave returns.
\newblock In \emph{Proceedings of the forty-fourth annual ACM symposium on
  Theory of computing}, pages 137--144, 2012.

\bibitem[Devanur et~al.(2013)Devanur, Jain, and Kleinberg]{DJK13}
Nikhil~R. Devanur, Kamal Jain, and Robert~D. Kleinberg.
\newblock Randomized primal-dual analysis of {RANKING} for online bipartite
  matching.
\newblock In Sanjeev Khanna, editor, \emph{Proceedings of the Twenty-Fourth
  Annual {ACM-SIAM} Symposium on Discrete Algorithms, {SODA} 2013, New Orleans,
  Louisiana, USA, January 6-8, 2013}, pages 101--107. {SIAM}, 2013.

\bibitem[Devanur et~al.(2019)Devanur, Jain, Sivan, and Wilkens]{DJSW19}
Nikhil~R. Devanur, Kamal Jain, Balasubramanian Sivan, and Christopher~A.
  Wilkens.
\newblock Near optimal online algorithms and fast approximation algorithms for
  resource allocation problems.
\newblock \emph{J. {ACM}}, 66\penalty0 (1):\penalty0 7:1--7:41, 2019.

\bibitem[Dvor{\'a}k and Henzinger(2014)]{DH14}
Wolfgang Dvor{\'a}k and Monika Henzinger.
\newblock Online ad assignment with an ad exchange.
\newblock In \emph{International Workshop on Approximation and Online
  Algorithms}, pages 156--167. Springer, 2014.

\bibitem[Esfandiari et~al.(2018)Esfandiari, Korula, and Mirrokni]{EKM2018}
Hossein Esfandiari, Nitish Korula, and Vahab Mirrokni.
\newblock Allocation with traffic spikes: Mixing adversarial and stochastic
  models.
\newblock \emph{ACM Transactions on Economics and Computation (TEAC)},
  6\penalty0 (3-4):\penalty0 1--23, 2018.

\bibitem[Feldman et~al.(2009)Feldman, Korula, Mirrokni, Muthukrishnan, and
  P{\'a}l]{FKMM09}
Jon Feldman, Nitish Korula, Vahab Mirrokni, Shanmugavelayutham Muthukrishnan,
  and Martin P{\'a}l.
\newblock Online ad assignment with free disposal.
\newblock In \emph{International workshop on internet and network economics},
  pages 374--385. Springer, 2009.

\bibitem[Feldman et~al.(2010)Feldman, Henzinger, Korula, Mirrokni, and
  Stein]{FHKMS10}
Jon Feldman, Monika Henzinger, Nitish Korula, Vahab~S Mirrokni, and Cliff
  Stein.
\newblock Online stochastic packing applied to display ad allocation.
\newblock In \emph{European Symposium on Algorithms}, pages 182--194. Springer,
  2010.

\bibitem[Huang et~al.(2020)Huang, Zhang, and Zhang]{HZZ20}
Zhiyi Huang, Qiankun Zhang, and Yuhao Zhang.
\newblock Adwords in a panorama.
\newblock In \emph{61st {IEEE} Annual Symposium on Foundations of Computer
  Science, {FOCS} 2020, Durham, NC, USA, November 16-19, 2020}, pages
  1416--1426. {IEEE}, 2020.

\bibitem[Kalyanasundaram and Pruhs(2000)]{KP00}
Bala Kalyanasundaram and Kirk Pruhs.
\newblock An optimal deterministic algorithm for online b-matching.
\newblock \emph{Theor. Comput. Sci.}, 233\penalty0 (1-2):\penalty0 319--325,
  2000.

\bibitem[Karp et~al.(1990)Karp, Vazirani, and Vazirani]{KVV90}
Richard~M Karp, Umesh~V Vazirani, and Vijay~V Vazirani.
\newblock An optimal algorithm for on-line bipartite matching.
\newblock In \emph{Proceedings of the twenty-second annual ACM symposium on
  Theory of computing}, pages 352--358, 1990.

\bibitem[Mehta(2013)]{Mehta13}
Aranyak Mehta.
\newblock Online matching and ad allocation.
\newblock \emph{Found. Trends Theor. Comput. Sci.}, 8\penalty0 (4):\penalty0
  265--368, 2013.

\bibitem[Mehta et~al.(2007)Mehta, Saberi, Vazirani, and Vazirani]{MSVV07}
Aranyak Mehta, Amin Saberi, Umesh Vazirani, and Vijay Vazirani.
\newblock Adwords and generalized online matching.
\newblock \emph{Journal of the ACM (JACM)}, 54\penalty0 (5):\penalty0 22--es,
  2007.

\bibitem[Mirrokni et~al.(2012)Mirrokni, Gharan, and Zadimoghaddam]{MGZ12}
Vahab~S. Mirrokni, Shayan~Oveis Gharan, and Morteza Zadimoghaddam.
\newblock Simultaneous approximations for adversarial and stochastic online
  budgeted allocation.
\newblock In Yuval Rabani, editor, \emph{Proceedings of the Twenty-Third Annual
  {ACM-SIAM} Symposium on Discrete Algorithms, {SODA} 2012, Kyoto, Japan,
  January 17-19, 2012}, pages 1690--1701. {SIAM}, 2012.

\end{thebibliography}

\appendix

\section{Online Vertex-weighted Matching with Surplus Supply} \label{app:uniform}

When there is no \adx node present, the online allocation problem degenerates to a classic online matching problem. Thus it is meaningful to ask the following fundamental question: when there is additional supply for an online matching problem, what is the optimal online algorithm? For online unweighted matching problem \cite{KVV90}, when there is no additional supply, the optimal online algorithm achieves a competitive ratio $1-\frac{1}{e}$. Such a result was extended to the vertex-weighted setting \cite{AGKM11} with the same optimal approximation ratio. In this section, we show that the algorithm in \cite{AGKM11} for the vertex-weighted problem can be extended to the setting with supply factor $f>1$, and shows that the competitive ratio improves to $f-fe^{-1/f}$.

\subsection{Online Algorithm}

We study the following vertex-weighted matching problem. There are $m$ advertisers and many online queries. Each advertiser $a$ demands $n_a$ queries, and has weight $c_a$ for any allocated query. Each online query can be allocated to some advertisers, such that there exists an offline allocation where each advertiser $a$ is matched to $fn_a$ impressions for some integer $f$. Unlike the main allocation problem with the presence of $\adx$ we study, we do not need to assume the demand of each advertiser is sufficiently large. The algorithm is similar to an algorithm \cite{AGKM11}, and the potential function used in \cite{AGKM11} needs to be modified to take into account the supply factor. The algorithm below is designed specifically for the problem where each advertiser $a$ demands $n_a=1$ impression. However, it can be generalized to arbitrary demand by reducing the problem with $n_a>1$ to the setting with $n_a=1$ for each advertiser, through splitting each advertiser $a$ to $n_a$ advertisers with demand 1, and the same set of demanding queries. Thus without loss of generality, we assume $n_a=1$ for each advertiser. Also, we note that in the special case of online unweighted matching, this algorithm is equivalent to the well-known \textsc{Ranking} algorithm of \cite{KVV90}, where the knowledge of $f$ is not needed.

\begin{theorem}
Given an online vertex-weighted matching problem with an integer supply factor $f$ and arbitrary demands, there exists a randomized online algorithm with competitive ratio $f-fe^{-1/f}$.
\end{theorem}

\begin{proof}[Proof sketch.]

The proof is almost identical to that used in \cite{AGKM11} with supply factor $f$, so we omit most of the details and only describe the differences. The algorithm is a perturbed version of the \textsc{Ranking} algorithm in \cite{KVV90} and the \textsc{Perturbed-Greedy} algorithm in \cite{AGKM11}.

\begin{algorithm}[htbp]
\caption{Optimal algorithm for vertex-weighted bipartite matching with additional supply}
\label{alg:vertex-weighted}
\textbf{Input:}{ Weight $c_a$ for each advertiser $a$, and supply factor $f$.}\\
\textbf{Preprocessing:} For each advertiser $a$, select $x_a\in[0,1]$ uniformly at random.\\
Define the function $\psi(x)=1-e^{-(1-x)/f}$.\\
 \For{each query arriving online}{
    Match the query to a matching advertiser $a$ with largest $c_a\psi(x_a)$. Break ties by advertiser id.
    }
\end{algorithm}

The only difference between Algorithm \ref{alg:vertex-weighted} and that of \cite{AGKM11} is that $\psi(x)=1-e^{-(1-x)/f}$ instead of $1-e^{-(1-x)}$. We now show how to modify their proofs to this additional-supply setting.

The choose of $x_a$ is equivalent to select a random integer $\sigma(a)\in[k]$ for each advertiser, with $k\to\infty$. The potential function is discretized to $\psi(i)=1-(1-\frac{1}{fk})^{-(k-i+1)}$ for each $i\in[k]$, and the algorithm chooses advertiser $a$ with highest $c_a\psi(i)$. 

We rewrite some of the definitions from \cite{AGKM11}.

\begin{definition}[Definition 7 in \cite{AGKM11}]\label{def7}
We say an advertiser $a$ is at position $t$, if $\sigma(a)=t$.

Let $Q_t$ be the set of all occurrences of matched vertices in the probability space:
\[Q_t=\{(\sigma,t,a):\sigma(a)=t\textrm{ and the vertex $a$ at position } t\textrm{ is matched in }\sigma\}.\]

Let $R_t$ be the set of all occurrences of unmatched vertices in the probability space:
\[R_t=\{(\sigma,t,a):\sigma(a)=t\textrm{ and the vertex $a$ at position } t\textrm{ is unmatched in }\sigma\}.\]

Let $x_t$ be the expected gain at $t$, over the random choice of $\sigma$. Then
\[x_t=\frac{\sum_{(\sigma,t,a)\in Q_t}c_a}{k^m}.\]

\end{definition}

The expected gain of the algorithm is $\alg_{\sigma}=\sum_t x_t$. The optimal gain at any position $t$ is $B=\frac{\opt}{k}=\frac{1}{k}\sum_{a}c_a$ since each advertiser $a$ is matched in the offline optimal allocation and appears at position $t$ with probability $\frac{1}{k}$. Then
\[B-x_t=\frac{\sum_{(\sigma,t,a)\in R_t}c_a}{k^m}.\]
\begin{definition}[Definition 8 in \cite{AGKM11}]\label{def8}
For any $\sigma$, let $\sigma_a^{i}\in[k]^m$ be obtained from changing the position of $a$ to $i$, i.e. $\sigma_a^i(a)=i$ and $\sigma_a^i(a')=\sigma(a')$ for $a'\neq a$.

\end{definition}

\begin{definition}[Definition 9 in \cite{AGKM11}]\label{def9}
For every $(\sigma,t,a)\in R_t$, define the set-valued charging map
\begin{eqnarray*}
    f_{map}(\sigma,t,a)&=&\{(\sigma_a^i,s,a'):1\leq i\leq k,\textrm{ the algorithm matches one of the }f\textrm{ impressions} \\
    &&\textrm{that gets allocated to advertiser }a\textrm{ in the offline matching with additional supply}\\
    &&\textrm{to }a'\textrm{ in }\sigma^i_a\textrm{ where }\sigma^i_a(a')=s\}.
\end{eqnarray*}

\end{definition}

\begin{observation}[Observation 2 in \cite{AGKM11}]\label{obs2}
For any $(\rho,s,a')\in f_{map}(\sigma,t,a)$, $(\rho,s,a')\in Q_s$.
\end{observation}

\begin{lemma}[Lemma 5 in \cite{AGKM11}]\label{lem5}
If the advertiser $a$ at position $t$ in $\sigma$ is unmatched by our algorithm, then for every $1 \leq i \leq k$,
the algorithm matches any impression which gets allocated to $a$ in the offline matching with additional supply to an advertiser $a'$ in $\sigma_a^i$ such that $\psi(t)c_a\leq\psi(\sigma_a^i(a'))c_{a'}$.
\end{lemma}

\begin{observation}[Observation 3 in \cite{AGKM11} with supply factor $f$]\label{obs3}
For any $(\sigma,t,a)\in R_t$, $f_{map}(\sigma,t,a)$ contains $fk$ values.
\end{observation}

\begin{definition}[Definition 10 in \cite{AGKM11}]\label{def10}
Let $S_t=\{(\sigma,t,a)\in R_t:(\sigma_a^{t-1},t-1,a)\not\in R_{t-1}\}$.
\end{definition}

\begin{claim}[Definition 11 and Claim 2 in \cite{AGKM11}]\label{cla2}
Let $\alpha_t=\frac{\sum_{(\sigma,t,a)\in S_t}c_a}{k^m}$. Then
\[x_t=B-\sum_{s\leq t}\alpha_s,\]
\[\textrm{Total loss}=\sum_{t}(B-x_t)=\sum_{t}(k-t+1)\alpha_t.\]
\end{claim}

\begin{claim}[Claim 3 in \cite{AGKM11}]\label{cla3}
For any $(\sigma,t,a)\in S_t$ and $(\rho,s,a')\in S_{s}$, if $(\sigma,t,a)$ is not identical to $(\rho,s,a')$, then $f_{map}(\sigma,t,a)$ and $f_{map}(\rho,s,a')$ are disjoint.
\end{claim}

Now we are ready to prove the main theorem.

\begin{theorem}[Theorem 6 in \cite{AGKM11} with supply factor $f$]\label{thm6}
As $k\to\infty$,
\[\textrm{total gain}=\sum_{t}x_t\geq(f-fe^{-1/f})\opt=(f-fe^{-1/f})\sum_ac_a.\]

\end{theorem}

\begin{proof}[Proof of Theorem~\ref{thm6} following the proof of Theorem 6 in \cite{AGKM11}]
Using Lemma~\ref{lem5} and Observation~\ref{obs3},
\[\psi(t)c_a\leq\frac{1}{fk}\sum_{(\sigma_a^i,s,a')\in f_{map}(\sigma,t,a)}\psi(s)c_{a'}.\]
Add the above equation for all $(\sigma,t,a)\in S_t$ for all $1\leq t\leq k$, then using Claim~\ref{cla3} and Observation~\ref{obs2} we have
\begin{eqnarray*}
\sum_{t}\psi(t)\frac{\sum_{(\sigma,t,a)\in S_t}c_a}{k^m}&\leq&\frac{1}{fk}\sum_{t}\psi(t)\frac{\sum_{(\sigma,t,a)\in Q_t}c_a}{k^m}\\
\sum_{t}\psi(t)\alpha_t&\leq&\frac{1}{fk}\sum_{t}\psi(t)x_t\\
&=&\frac{1}{fk}\sum_{t}\psi(t)\left(B-\sum_{s\leq t}\alpha_s\right).
\end{eqnarray*}
Here the second line is by Claim~\ref{cla2} and Definition~\ref{def7}. The third line is by Claim~\ref{cla2}. By rearranging the above inequality we have
\begin{equation*}
    \sum_t\alpha_t\left(\psi(t)+\frac{\sum_{s\geq t}\psi(s)}{fk}\right)\leq\frac{B}{fk}\sum_t\psi(t).
\end{equation*}
For $\psi(t)=1-(1-\frac{1}{fk})^{k-t+1}$, observe that $\psi(t)+\frac{\sum_{s\geq t}\psi(s)}{fk}\geq\frac{k-t+1}{fk}$, and $\sum_t\psi(t)=k(1-f+fe^{-1/f})$ when $k\to\infty$. Using Claim~\ref{cla2} we have
\begin{eqnarray*}
\textrm{total loss}&=&\sum_t(B-x_t)=\sum_t(k-t+1)\alpha_t\\
&\leq&fk\sum_t\alpha_t\left(\psi(t)+\frac{\sum_{s\geq t}\psi(s)}{fk}\right)\\
&\leq&B\sum_t\psi(t)=kB(1-f+fe^{-1/f})=(1-f+fe^{-1/f})\opt.
\end{eqnarray*}
Thus the total gain of the algorithm is at least $(f-fe^{-1/f})\opt$.
\end{proof}

\end{proof}

\subsection{Tightness (upper bound).}
Next, we prove that the online algorithm described in the previous section is tight by arguing that no randomized online algorithm can get a better competitive ratio.

First, we recall the instance that we repeatedly used for tightness results throughout the paper:
\kvv*

We use this instance to show:
\begin{theorem}
There exists an instance of the unweighted matching problem with supply factor $\f$, for which no online algorithm can obtain a competitive ratio better than $\f-\f e^{-1/\f}$.
\end{theorem}
\begin{proof}
By Yao's min-max principle, to show that no randomized online algorithm can obtain a competitive ratio better than $f-fe^{-1/f}$ for adversarial queries, we only need to prove that no deterministic online algorithm can obtain a competitive ratio better than $f-fe^{-1/f}$ for stochastic queries. Consider the instance in Example \ref{ex:kvv-general}.

For any advertiser $j$ such that $\pi(j)\geq i$, $\mathbb{E}[\# \textrm{impressions allocated to j in phase i}] \leq \frac{fn}{m-i+1}$, since by definition of the instance there are at most $fn$ impressions allocated to each of the advertisers with $\pi(j)\geq i$, and $\pi$ is a random permutation. Thus the expected number of impressions allocated to any advertiser $j$ is $\min(n,\sum_{i\leq \pi(j)}\frac{fn}{m-i+1})$. The expected reward of any deterministic algorithm when is bounded by,
\begin{eqnarray*}
\sum_{j\in[m]}\min(n,\sum_{i\leq \pi(j)}\frac{fn}{m-i+1})
&=&\sum_{j\in[m]}\min(n,\sum_{i\leq j}\frac{fn}{m-i+1})\\
&\leq&\sum_{j=1}^{m(1-e^{-1/f})}\sum_{i\leq j}\frac{fn}{m-i+1}+\sum_{j>m(1-e^{-1/f})}n\\
&\leq&\sum_{j=1}^{m(1-e^{-1/f})}fn\ln\left(\frac{m}{m-j}\right)+mne^{-1/f}\\
&=&fn\ln\left(\frac{m^{m(1-e^{-1/f})}}{m!/(me^{-1/f})!}\right)+mne^{-1/f}\\
&\leq&fn\Bigg(m(1-e^{-1/f})\ln m-(m\ln m-m+\frac{1}{2}\ln m)\\
& &+(1+me^{-1/f}\ln(me^{-1/f})-me^{-1/f}+\frac{1}{2}\ln (me^{-1/f}))\Bigg)+mne^{-1/f}\\
&=&fnm\left(1-e^{-1/f}+\frac{1+1/2f}{m}\right)=nm\left(f-fe^{-1/f}+O\left(\frac{1}{m}\right)\right).
\end{eqnarray*}

Here the third line is by $\sum_{i=j+1}^{m}\frac{1}{i}\leq \ln(\frac{m}{j})$. The inequality in the fifth line follows by Stirling's formula that states $\ln k!-(k\ln k-k+\frac{1}{2}\ln k)\in[0,1]$ for any positive integer $k$. Since in the optimal offline allocation, each advertiser can get allocated $n$ impressions, thus the offline optimal value is $n m$. Thus the competitive ratio of any online algorithm is at most $f-fe^{-1/f}$ for large $m$.
\end{proof}

\section{Deferred Proof from Section \ref{sec:DP}}\label{app:DP}
In this appendix section, we briefly describe a dynamic programming approach for computing the thresholds efficiently.
\begin{proof}[Proof of Theorem~\ref{thm:DP}]
The problem can be solved via the following dynamic program. Let $g[i, x, y]$ denote the maximum of $\sum_{u \leq i} \sum_{j=s_{u-1}t}^{s_{u}t}  \beta^*_j (c- E_F[r \mid r \leq r_{d+1-u}])$, where $x$ stores the $\beta^*_{s_{i}t}$ value, when $y$ stores the value of $s_i$. 
Observe that the objective 
\begin{equation*}
    \max_{\mathbf{s}}\lb_{m,N}(s_1,...,s_d)=-cN+\sum_{u=1}^{d}fN(q_{u}-q_{u-1})r_u+\max_{0\leq x\leq 1}g[d,x,1].
\end{equation*}

Then it suffices to show that we can efficiently solve $\max_{0\leq x\leq 1}g[d,x,1]$ with small error. We can write down the following recurrence formula for $f$:
\begin{eqnarray}
    & &g[i,x,y]\nonumber\\
    &=&\max_{x=\beta'_{s_i},y=s_i}\sum_{u=1}^{i}\sum_{j=s_{u-1}t+1}^{s_ut}\beta^*_j(c-\E_{F}[r|r\leq r_{d+1-u}])\nonumber\\
    &=&\max_{x=\beta'_{s_i},y=s_i}\left(\sum_{u=1}^{i-1}\sum_{j=s_{u-1}t+1}^{s_ut}\beta^*_j(c-\E_{F}[r|r\leq r_{d+1-u}]+\sum_{j=s_{i-1}t+1}^{s_it}\beta^*_j(c-\E_{F}[r|r\leq r_{d+1-i}])\right)\nonumber\\
    &=&\max_{\substack{y'<y\\x'=x(1-\frac{1}{q_{d+1-i}})^{-(yt-y't)}}}\bigg(g[i-1,x',y']\nonumber\\
    & &+\sum_{j=y't+1}^{yt}x'\left(1-\frac{1}{q_{d+1-i}tf}\right)^{j-y't}(c-\E_F[r|r\leq r_{d+1-i}])\bigg)\label{eqn:recur-f}.
\end{eqnarray}
Here the first equation is the definition of $g$; the second equation is by separating the last term in the sum out; the last equation is by observing that the first term in the second equation can be expressed by $g$. However, notice that both $x$ and $y$ are defined to be real values in $[0,1]$, thus we need to discretize the space of $x$ and $y$ in order to solve the recurrence efficiently by dynamic program. 

To show that $y$ can be discretized to multiples of $\epsilon$, it suffices to show the following lemma.
\begin{lemma}\label{lem:discretize-y}
For any $0\leq s_1\leq\cdots\leq s_d=1$ and each $1\leq i\leq d$, let $\hat{s}_i$ be the largest multiple of $\epsilon$ that is no larger than $s_i$. Then $\lb_{m,N}(s_1,\cdots,s_d)\leq \lb_{m,N}(\hat{s}_1,\cdots,\hat{s}_d)+O(cN\epsilon)$. 
\end{lemma}
\begin{proof}

It suffices to show that for any $1\leq j\leq t$, $\beta_j'$ does not change too much when $s_1,\cdots,s_d$ are rounded to $s_1',\cdots,s_d'$.
Let $\delta_{u}=1-\frac{1/q_{d+1-u}}{tf}$, for each $1\leq u\leq d$. Then $\delta_1\geq\delta_2\geq\cdots\geq\delta_d$, and for any $s_{u-1}t+1<j\leq s_{u}t+1$,
\begin{eqnarray*}
\beta^*_j(s_1,\cdots,s_d)&=&\delta_1^{s_1 t-s_0 t}\cdots\delta_{u-1}^{s_{u-1}t-s_{u-2}t}\delta_u^{j-s_{u-1}t-1}\\
&=&\delta_u^j\left(\frac{\delta_{u-1}}{\delta_u}\right)^{s_{u-1}t}\cdot\left(\frac{\delta_{u-2}}{\delta_{u-1}}\right)^{s_{u-2}t}\cdot\cdots\cdot\left(\frac{\delta_{1}}{\delta_{2}}\right)^{s_{1}t}\\
&=&\delta_u^j\left(\frac{\delta_{u-1}}{\delta_u}\right)^{s_{u-1}t-\hat{s}_{u-1}t+\hat{s}_{u-1}t}\cdot\left(\frac{\delta_{u-2}}{\delta_{u-1}}\right)^{s_{u-2}t-\hat{s}_{u-2}t+\hat{s}_{u-2}t}\cdot\cdots\cdot\left(\frac{\delta_{1}}{\delta_{2}}\right)^{s_{1}t-\hat{s}_{1}t+\hat{s}_{1}t}\\
&=&\beta^*_j(\hat{s}_1,\cdots,\hat{s}_d)\left(\frac{\delta_{u-1}}{\delta_u}\right)^{s_{u-1}t-\hat{s}_{u-1}t}\cdot\left(\frac{\delta_{u-2}}{\delta_{u-1}}\right)^{s_{u-2}t-\hat{s}_{u-2}t}\cdot\cdots\cdot\left(\frac{\delta_{1}}{\delta_{2}}\right)^{s_{1}t-\hat{s}_{1}t}\\
&<&\beta^*_j(\hat{s}_1,\cdots,\hat{s}_d)\left(\frac{\delta_{u-1}}{\delta_u}\right)^{\eps t}\cdot\left(\frac{\delta_{u-2}}{\delta_{u-1}}\right)^{\eps t}\cdot\cdots\cdot\left(\frac{\delta_{1}}{\delta_{2}}\right)^{\eps t}\\
&=&\beta^*_j(\hat{s}_1,\cdots,\hat{s}_d)\left(\frac{\delta_{1}}{\delta_{u}}\right)^{\eps t}\\
&=&\beta^*_j(\hat{s}_1,\cdots,\hat{s}_d)(1+O(\eps)).
\end{eqnarray*}
Since $\beta^*_j(s_1,\cdots,s_d)<1$, we have $\beta^*_j(s_1,\cdots,s_d)-\beta^*_j(\hat{s}_1,\cdots,\hat{s}_d)=O(\eps)$, thus $\lb_{m,N}(s_1,\cdots,s_d)\leq \lb_{m,N}(\hat{s}_1,\cdots,\hat{s}_d)+O(cN\epsilon)$.

\end{proof}
To discretize $x$, consider calculating the recursive formula by  
\begin{eqnarray*}
    g[i,x,y]&=&\max_{\substack{y'\leq y\\x'=\frac{\eps}{d}\lfloor\frac{d}{\eps}x(1-\frac{1}{q_{d+1-i}})^{-(yt-y't)}\rfloor}}\bigg(g[i-1,x',y']\nonumber\\
    & &+\sum_{j=y't+1}^{yt}x'\left(1-\frac{1}{q_{d+1-i}tf}\right)^{j-y't}(c-\E_F[r|r\leq r_{d+1-i}])\bigg)
\end{eqnarray*}
instead of \eqref{eqn:recur-f}. In other words, when we need to calculate the value of $\beta_{s_{i-1}}$ and check $g[i-1,\beta_{s_{i-1}},y']$, we first round $\beta_{s_{i-1}}$ to $\hat{\beta}_{s_{i-1}}$ that is the closest multiple of $\frac{\epsilon}{d}$, then call $g[i-1,\hat{\beta}_{s_{i-1}},y']$. For each $i$ the discretization incurs an additive error of $N\cdot\frac{\epsilon}{d}\cdot c$, which leads to an overall error of $cN\epsilon$. By setting $\epsilon=\frac{1}{m}$, there is an error of $O(\frac{cN}{m})$ for discretizing $x$ to multiples of $\frac{1}{md}$ and $y$ to multiples of $\frac{1}{m}$. The recursive formula can be efficiently solved via dynamic program with running time $O(d^2m^3)$, since there are $d\cdot md\cdot m=m^2d^2$ entries in the discretized table, while each time the max operator calls $O(m)$ values of $y'\leq y$. 
\end{proof}

\section{Other Deferred Proofs from Section \ref{sec:general_alg}} \label{app:general}
\subsection{Proof of Claim \ref{claim:beta_general}}
\begin{proof}[Proof sketch.]
Observe that the optimal solution of the above LP is when all nontrivial inequalities become equalities. This is because in an optimal solution, for the first constraint $j$ being a strict equality, if $j\not\in\{s_1 t,s_2t\cdots,s_d t\}$, then $\beta_{j+1}\leftarrow\beta_{j+1}-\epsilon$, $\beta_{j+2}\leftarrow\beta_{j+2}+\epsilon$ is a new feasible solution with the objective staying the same, while the $jth$ constraint can become tight; if $j=s_u t$ for some $u\in[d]$, then $\beta_{j+1}\leftarrow\beta_{j+1}-q_{d+1-u}\epsilon$, $\beta_{j+2}\leftarrow\beta_{j+2}+q_{d-u}\epsilon$ is a new feasible solution with the objective decrease by $(c-\E_{F}[r|r\leq r_{d+1-u}])q_{d+1-u}\epsilon-(c-\E_{F}[r|r\leq r_{d-u}])q_{d-u}\epsilon=(q_{d+1-u}-q_{d-u})(c-r_{d+1-u})\geq 0$ which is a non-negative value, while the $jth$ constraint can become tight. Repeat such process we can let all inequalities become tight, while the objective remains optimal.
\end{proof}

\subsection{Proof of Claim \ref{lem:negativelp_thresholds}}
\begin{proof}
For any optimal $\mathbf{y}$, suppose by way of contradiction $0<y_{iu}<\frac{1}{m-i+1}$, $0<y_{i'u}<\frac{1}{m-i'+1}$ for $i<i'$. Then setting $y_{iu}\leftarrow y_{iu}+\epsilon$, $y_{i'u}\leftarrow y_{i'u}-\epsilon$ for small enough $\epsilon$ leads to a new feasible solution since all constraints are still feasible. Furthermore, in the objective function $y_{iu}$ has coefficient $(m-i+1)\frac{fN}{m}(q_u-q_{u-1})(c-r_u)>(m-i'+1)\frac{fN}{m}(q_u-q_{u-1})(c-r_u)$ which is the coefficient of $y_{i'u}$. Thus after perturbing $\mathbf{y}$ we get a larger objective value, which contradicts the assumption of $\mathbf{y}$ being optimal.

For any optimal solution $\mathbf{y}$, if $z_u< z_{u'}$ for $u<u'$, then for $i=z_{u'}$, $y_{iu}<\frac{1}{m-i+1}$, while $y_{iu'}=\frac{1}{m-i+1}$. Then setting $y_{iu}\leftarrow y_{iu}+\frac{1}{q_u-q_{u-1}}\epsilon$, $y_{iu'}\leftarrow y_{iu'}-\frac{1}{q_{u'}-q_{u'-1}}\epsilon$ for small enough $\epsilon$ leads to a new feasible solution since all constraints are still feasible. Furthermore, the increase of the objective due to $y_{iu}$ is $(m-i+1)\frac{fN}{m}(c-r_u)>(m-i+1)\frac{fN}{m}(c-r_{u'})$ which is the decrease of the objective due to $y_{iu'}$. Thus after perturbing $\mathbf{y}$ we get a larger objective value, which contradicts the assumption of $\mathbf{y}$ being optimal.
\end{proof}

\subsection{Proof of Claim \ref{clm:lbub}}
\begin{proof} $s_{u-1}t+1<j\leq s_{u}t+1$ for $1\leq u\leq d$,
\begin{eqnarray*}
    \beta_j^*&=&\frac{N}{t}\left(1-\frac{1/q_d}{tf}\right)^{s_1 t-s_0 t}\cdots\left(1-\frac{1/q_{d+2-u}}{tf}\right)^{s_{u-1}t-s_{u-2}t}\left(1-\frac{1/q_{d+1-u}}{tf}\right)^{j-s_{u-1}t-1}\\
    &=&\frac{N}{t}\exp\left(\frac{s_1-s_0}{fq_d}+\frac{s_2-s_1}{fq_{d-1}}+\cdots+\frac{s_{u-1}-s_{u-2}}{fq_{d+2-u}}\right)\left(1-\frac{1/q_{d+1-u}}{tf}\right)^{j-s_{u-1}t-1}.
\end{eqnarray*}
Then since $t \to \infty$ we have,
\begin{eqnarray}
    \sum_{j=s_{u-1}t+1}^{s_ut}\beta_j^*&=&\frac{N}{t}\exp\left(\frac{s_1-s_0}{fq_d}+\frac{s_2-s_1}{fq_{d-1}}+\cdots+\frac{s_{u-1}-s_{u-2}}{fq_{d+2-u}}\right)q_{d+1-u}tf\left(1-\left(1-\frac{1/q_{d+1-u}}{tf}\right)^{s_ut-s_{u-1}t}\right)\nonumber\\
    &=&Nfq_{d+1-u}\exp\left(\frac{s_1-s_0}{fq_d}+\frac{s_2-s_1}{fq_{d-1}}+\cdots+\frac{s_{u-1}-s_{u-2}}{fq_{d+2-u}}\right)\left(1-\frac{s_u-s_{u-1}}{fq_{d+1-u}}\right)\label{eqn:sumofbeta}
\end{eqnarray}

Define $C=-cN+\sum_{u=1}^{d}fN(q_u-q_{u-1})r_u$ to be the constant that appears in definitions of both $\ub$ and $\lb$ functions, and $\gam_{u}=\frac{s_u-s_{u-1}}{fq_{d+1-u}}$ for each $u\in[d]$.

\begin{eqnarray*}
& \lb_{m,N}(s)&=C+\sum_{u=1}^{d}(c-\E_{\dist}[r|r\leq r_{d+1-u}])Nfq_{d+1-u}\gam_1\gam_2\cdots\gam_{u-1}\left(1-\gam_u\right)\\
&=&C+\sum_{u=1}^{d}\left(c-\frac{(q_1-q_0)r_1+\cdots+(q_{d+1-u}-q_{d-u})r_{d+1-u}}{q_{d+1-u}}\right)Nfq_{d+1-u}\gam_1\gam_2\cdots\gam_{u-1}\left(1-\gam_u\right)\\
&=&C+Nf\sum_{u=1}^{d}\Big((q_1-q_0)(c-r_1)+\cdots+(q_{d+1-u}-q_{d-u})(c-r_{d+1-u})\Big)\gam_1\gam_2\cdots\gam_{u-1}\left(1-\gam_u\right)\\
&=&C+Nf\sum_{u=1}^{d}(q_{d+1-u}-q_{d-u})(c-r_{d+1-u})\sum_{j=1}^{d+1-u}\gam_1\cdots\gam_{j-1}\left(1-\gam_j\right)\\
&=&C+Nf\sum_{u=1}^{d}(q_{d+1-u}-q_{d-u})(c-r_{d+1-u})(1-\gam_1\cdots\gam_{d+1-u})\\
&=&C+Nf\sum_{u=1}^{d}(q_{u}-q_{u-1})(c-r_{u})(1-\gam_1\cdots\gam_{u})=\ub_{m,N}(s).
\end{eqnarray*}
Here the first equality is by applying formula \eqref{eqn:sumofbeta} to Optimization Problem~\ref{def:lb}; 
the second equality is by the definition of $\E_{\dist}[r|r\leq r_{d+1-u}]$; 
the third equality is by $cq_{d+1-u}=c(q_1-q_0)+c(q_2-q_1)+\cdots+c(q_{d+1-u}-q_{d-u})$; 
the fourth equality is by regrouping the sum as a linear function of $(q_{d+1-u}-q_{d-u})(c-r_{d+1-u})$;
the fifth equality is by resolving the telescoping sum; 
the sixth equality is by replacing the iteration variable $u$ with $d+1-u$; 
the last equality is by the definition of $\ub_{m,N}(s)$ in Optimization Problem~\ref{def:negative_opt}. This finishes the proof of Claim~\ref{clm:lbub}.

\end{proof}

\section{Discussion about Competitive Ratio}
\label{sec:appendixcompetitiveratio}
In this section, we discuss the exact competitive ratio of the algorithms we proposed, i.e. Algorithm~\ref{alg:binary} for binary \adx distribution $\dist$, and Algorithm~\ref{alg:general} for general \adx distributions.

\subsection{Competitive Ratio of Algorithm~\ref{alg:binary} for Binary \adx Distribution}
\label{sec:binary_competitive}

In Section~\ref{sec:binary} we characterized the reward achieved by the algorithm. To calculate the competitive ratio, we need to characterize the optimal offline reward we can get from the instance. The following theorem applies to general \adx distribution $\dist$.

\begin{theorem}\label{thm:opt-general}
Consider any instance with total demand $N$ for all advertisers with \adx distribution $\dist$ (and cumulative density function $F$) , where each advertiser $a$ has sufficiently large demand $n_a$. Then if there are exactly $fN$ online queries, while there exists an offline matching such that each advertiser $a$ is matched to exactly $fn_a$ impressions, the optimal offline reward is
\begin{equation*}
\opt=(f-1)N\E_{x\sim \dist}[x|x\geq F^{-1}(\frac{1}{f})].
\end{equation*}
\end{theorem}

\begin{proof}[Proof of Theorem \ref{thm:opt-general}]
By the definition of supply factor $f$, there exists a matching such that each advertiser $a$ is matched to a set $I_a$ of $fn_a$ impressions. In an optimal offline allocation, $N=\sum_a n_a$ impressions are allocated to the contract, while $(f-1)N$ impressions are allocated to \adx. Then the total reward from the impressions are upper bounded by the sum of the $(f-1)N$ largest \adx reward of all impressions. Let $F$ be the cumulative density function of distribution $\dist$. Then the optimal offline reward is
\begin{eqnarray*}
\opt&\leq&\E_{X_1,\cdots,X_{fN}\sim\dist}[\textrm{Sum of the largest }(f-1)N\textrm{ variables in }X_1,\cdots,X_{fN}]\\
&\leq& (f-1)N\E_{x\sim \dist}[x|x\geq F^{-1}(\frac{1}{f})].
\end{eqnarray*}
On the other hand, consider an allocation that only allocates impressions in $I_a$ to either advertiser $a$ or \adx. In particular, impressions with the smallest $n_a$ \adx reward are allocated to $a$, and the rest of the impressions are allocated to \adx. Then we have
\begin{eqnarray*}
\opt&\geq&\sum_{a}\E_{X_1,\cdots,X_{fn_a}\sim\dist}[\textrm{Sum of the largest }(f-1)n_a\textrm{ variables in }X_1,\cdots,X_{fn_a}]\\
&\geq& \sum_{a}(f-1)n_a\E_{x\sim \dist}[x|x\geq F^{-1}(\frac{1}{f})](1-o(1))=(1-o(1))(f-1)N\E_{x\sim \dist}[x|x\geq F^{-1}(\frac{1}{f})].
\end{eqnarray*}
The second line holds since when $n_a\to\infty$, with high probability the empirical distribution of the \adx weights of the impressions has a negligible distance from the true \adx distribution $\dist$ (using a standard concentration bound, we see that the deviation from the expectation is $\tilde{O}(\frac{1}{\sqrt{n_a}})$). The theorem follows by combining the lower bound and the upper bound of $\opt$.
\end{proof}

Now we are ready to analyze the competitive ratio.

\begin{theorem} \label{thm:binaryratio}
For a penalty $c$, supply factor $f \geq 1$, and a binary reward distribution with parameters $q$ and $r$, where $r \leq c$, Algorithm \ref{alg:binary} has a competitive ratio depending on $q,c,r,f$ as follows:
\begin{align*}
\begin{cases} 
    \frac{(1-1/f)-(1-r/c) ^{1-q}e^{-1/f}}{(1-q)}\cdot\frac{c}{r} &\mbox{if } q > 1/f \textrm{ and } q\geq\frac{1}{f\ln\frac{c}{c-r}}; \\
 (1-\frac{(1-r/c) ^{1-q}e^{-1/f}}{(1-1/f)})\cdot\frac{c}{r} & \mbox{if } q \leq 1/f \textrm{ and } q\geq\frac{1}{f\ln\frac{c}{c-r}};\\
  \frac{(1-1/f)+(1-q)(1-r/c)- qe^{-\frac{1}{qf}}}{(1-q)}\cdot\frac{c}{r} &\mbox{if } q > 1/f \textrm{ and } q<\frac{1}{f\ln\frac{c}{c-r}}; \\
 (1+\frac{(1-q)(1-r/c)- qe^{-\frac{1}{qf}}}{(1-1/f)})\cdot\frac{c}{r} & \mbox{if } q \leq 1/f \textrm{ and } q<\frac{1}{f\ln\frac{c}{c-r}}. 
\end{cases}
\end{align*}
\end{theorem}

We note that there are settings in which the competitive ratio of our algorithm (and more generally \textit{any online algorithm}) may be \textit{negative}. This also underscores the importance of the supply factor $f$ in these types of penalty settings. 

\begin{proof}

Recall that in the binary \adx distribution, with probability $q$ the \adx reward is $0$ for an impression and with probability $1-q$ it is $r$. By Theorem~\ref{thm:opt-general}, the optimal reward in the binary case can be written as
\begin{equation*}
    \opt=\begin{cases} 
Nf(1-q)r  &\mbox{if } q > 1/f; \\
Nf(1-1/f)r & \mbox{if } q \leq 1/f. 
\end{cases}
\end{equation*}
Since in the first case all the impressions with \adx value $0$ are assigned to contracts and since $q>1/f$ all contracts are satisfied and there is no penalty. We get expected reward of $(1-q) r$. In the second case an impression is allocated to \adx only after all contracts are satisfied, and hence $N$ (or a $1/f$ fraction of all queries) go to contracts and the remaining goes to \adx. By comparing each of these cases with the value of $\alg$ in the cases considered in Claim \ref{claim:threshold}, we get a competitive ratio minimizing the four cases described in the theorem statement.  

\end{proof}

\subsection{Worst-case Competitive Ratio for Ad Exchange Distributions with Fixed Mean }\label{sec:comp_ratio}
The exact competitive ratio for general \adx distribution is hard to describe, and as discussed in Section~\ref{sec:DP}, even computing the optimal thresholds are not straightforward. Surprisingly, we are able to precisely characterize the \adx distribution $\dist$ with the worst competitive ratio over all reward distributions with the same mean.

In particular, let $\mathcal{F}_{\mu}$ be the class of all value distributions $\dist$ such that the mean of distribution $\dist$ is $\mu$.
We show that among all \adx reward in $\mathcal{F}_{\mu}$, the optimal competitive ratio is minimized when $\dist$ is a binary distribution, i.e. a distribution with support size 2. To prove this, we analyze the maximum reward obtained by any online algorithm on Example~\ref{ex:kvv-general}.
Note also that we have already shown in Theorem~\ref{thm:general-hardness} that there is always a threshold -based algorithm that achieves the optimal reward, and hence we can restrict our attention to threshold-based algorithms. 
\begin{theorem}\label{thm:worstcase_dist} Consider the class $\mathcal{F}_{\mu}$ of \adx distribution with mean $\mu$, and let $\ratio(\dist)$ be the reward of best online algorithm on Example~\ref{ex:kvv-general} on a distribution $\dist \in \mathcal{F}_{\mu}$. Then $\ratio(\dist)$ is minimized when $\dist$ is one of the following:
\begin{itemize}
\item A fixed distribution with reward $\mu$.
\item A (binary) distribution with parameters $\{(r_1=0, q_1=1/f), (r_2=\frac{f}{f-1}\mu, q_2=\frac{f-1}{f}) \}$.
\item A (binary) distribution with parameters $\{(r_1=c, q_1=\frac{f-1}{f}), (r_2=f\mu - (f-1)c, \frac{f-1}{f}, q_2=1/f) \}$
\end{itemize}
\end{theorem}

This theorem allows us to characterize the worst case competitive ratio for a general distribution, with the competitive ratio of a fixed or a binary reward distribution as determined in Section \ref{sec:binary}. The exact competitive ratio can be computed by combining Theorem \ref{thm:worstcase_dist} and Theorem \ref{thm:binaryratio}. The combination of these two theorems, involves many cases depending on the relations between $r,c,f, \mu$, which we do not mention here. 

\begin{proof}
In Section~\ref{sec:general-hardness}, we showed that the optimal algorithm for Example~\ref{ex:kvv-general} has the following form: For each group $G_i$ of queries that arrives, the algorithm sets a threshold $t_i$ such that for any impression with \adx reward $\leq t_i$, the impression is allocated to the advertiser with the lowest satisfaction ratio; otherwise, the impression is allocated to \adx. Let $\mathcal{A}$ be the class of algorithms with such a form. An equivalent interpretation is that the algorithm allocates $x_i=F(t_i)$ fraction of all queries in $G_i$ to the advertiser with lowest satisfaction ration, and $(1-x_i)$ fraction of all the impressions to \adx. Note that each advertiser is matched with at most $n$ impressions. Also, using similar reasoning as in Section \ref{sec:general_alg}, we know $\frac{fx_i}{m-i+1}n$ impressions in group $G_i$ are allocated to available advertisers in that group. Thus the following constraint holds for any algorithm $A\in\mathcal{A}$ and advertisers available in each group $G_i$:
\begin{equation*}
    \sum_{i=1}^{m}\frac{fx_i}{m-i+1}n\leq n.
\end{equation*}

Now we can compute the expected reward $\alg$ of algorithm $A$ defined by $(x_1,\cdots,x_m)$. Firstly, the penalty of not satisfying the contracts is $-mn_ac$ when impressions are not allocated, and we allocate $\sum_{i=1}^{m}fx_in$ impressions in total. Thus the total reward (or penalty) contributed to the objective by contract advertisers is $(\sum_{i=1}^{m}fx_in-mn)c$ in total. Secondly, the \adx reward contributed from impressions in $G_i$ is w.h.p.\footnote{Note that we can used this expected values as a high-probability reward bound for the exact same reasons as we showed in Section \ref{sec:binary_competitive}.}~$f(1-x_i)n\E_{r\sim F}[r|r\geq F^{-1}(x_i)]$ for each phase $i$. This is because there are $fn$ impressions in $G_i$, and the algorithm allocates $(1-x_i)$ fraction of impressions with highest valued \adx reward to \adx. Combining the contribution from contracts and \adx, we have the following objective for any such algorithm:
\begin{equation}\label{eqn:fixedmeanalg}
    \frac{1}{n}\alg=\left(\sum_{i=1}^{m}fx_i-m\right)c+\sum_{i=1}^{m}f(1-x_i)\E_{r\sim F}[r|r\geq F^{-1}(x_i)].
\end{equation}
Next, we bound the optimal offline reward as follows.  We can use Theorem \ref{thm:opt-general} to get the following:
\[\frac{1}{n}\opt=(f-1)m\E_{r\sim F}\left[r\Big|r\geq F^{-1}\left(\frac{1}{f}\right)\right].\]
Recall, that this follows from the fact that an optimal allocation can choose to allocate $(f-1)mn$ queries with high \adx weight to \adx and $mn$ queries with \adx weight in the lowest $\frac{1}{f}$ quantile to satisfy the contract of all the advertisers.

Let $\mu_1:=\E_{r\sim F}[r|r<F^{-1}(1/f)]$ and $\mu_2:=\E_{r\sim F}[r|r\geq F^{-1}(1/f)]$ be the expected value of bottom $\frac{1}{f}$ quantile and top $1-\frac{1}{f}$ quantile of distribution $\dist$ respectively. 
Then we can rewrite \begin{equation}\label{eqn:fixedmeanopt}
    \frac{1}{n}\opt=(f-1)m\mu_2.
\end{equation}
Observe that for $x_i\geq\frac{1}{f}$,
\begin{equation}\label{eqn:topquantile}
    (1-x_i)\E_{r\sim F}[r|r\geq F^{-1}(x_i)]\geq (1-x_i)\mu_2,
\end{equation}
and the equality holds when $x_i=\frac{1}{f}$, or the top $1-\frac{1}{f}$ quantile of $\dist$ has a fixed value $\mu_2$, i.e. $\Pr_{r\sim F}[r=\mu_2]\geq1-\frac{1}{f}$ if $x_{i}>\frac{1}{f}$;
for $x_i<\frac{1}{f}$,
\begin{equation}\label{eqn:bottomquantile}
    (1-x_i)\E_{\sim F}[r|r\geq F^{-1}(x_i)]\geq (\frac{1}{f}-x_i)\mu_1+(1-1/f)\mu_2,
\end{equation}
and the equality hold only when the bottom $\frac{1}{f}$ quantile of $\dist$ has a fixed value $\mu_1$, i.e. $\Pr_{r\sim \dist}[r=\mu_1]\geq\frac{1}{f}$.
Thus for any algorithm $A\in\mathcal{A}$ defined by allocation probability $(x_1,\cdots,x_m)$, apply the above inequalities \eqref{eqn:topquantile} and \eqref{eqn:bottomquantile} to formulas \eqref{eqn:fixedmeanalg} and \eqref{eqn:fixedmeanopt}, we get the competitive ratio of algorithm $A$ is
\begin{equation}\label{eqn:fixedmeanratio}
    \frac{\alg}{\opt}=\frac{\left(\sum_{i=1}^{m}fx_i-m\right)c+\sum_{i=1}^{m}f(1-x_i)\E[r\sim F|r\geq F^{-1}(x_i)]}{(f-1)m\mu_2}\geq\frac{-X+Y\mu_1+Z\mu_2}{\mu_2}
\end{equation}
for coefficients $X,Y,Z$ that are derived from \eqref{eqn:topquantile} and \eqref{eqn:bottomquantile} that only depend on $x_1,\cdots,x_m$. Since $\mu_1=f\mu-(f-1)\mu_2$, the right hand side of the inequality \eqref{eqn:fixedmeanratio} is a linear function of $\frac{1}{\mu_2}$, thus monotone with respect to $\mu_2$. The minimum can be achieved when $\mu_2$ is either minimized or maximized (depending on the sign of the coefficient). 
Since $\mu_2\geq\mu$, thus the minimum value of $\mu_2$ is $\mu$, and in this case $\mu_1=1$.
If $\frac{f}{f-1}\mu\leq c$, then $\mu_2$ has maximum value $c$, and in this case $\mu_1=f\mu-(f-1)c$; otherwise, $\mu_2$ has maximum value $\frac{f}{f-1}\mu$, and in this case $\mu_1=0$.
The minimum of the right hand side of inequality \eqref{eqn:fixedmeanratio} can be achieved in one of the above three cases. Notice that the equality of \eqref{eqn:fixedmeanratio} can hold when $\dist$ is a binary distribution with support $\mu_1$ and $\mu_2$ with corresponding probability $\frac{1}{f}$ and $1-\frac{1}{f}$. 
Thus for any algorithm $A\in \mathcal{A}$, the reward of $A$ for Example~\ref{ex:kvv-general} is minimized when $\dist$ is a binary distribution, among all distributions with fixed mean $\mu$.
\end{proof}

\end{document}